\documentclass[reqno,10pt]{amsart}
\usepackage{amsmath}
\usepackage{amsthm}
\usepackage{amssymb}
\usepackage[mathscr]{euscript} 
\usepackage{bm}
\usepackage{pgf}
\usepackage{color,xcolor}
\usepackage{graphicx}
\usepackage{comment}
\usepackage{mathtools}
\usepackage{soul}

\usepackage{a4wide}

\newcommand\be[1]{\begin{equation}\label{#1}}
\newcommand\ee{\end{equation}}
\newcommand\ba[1]{\begin{align}\label{#1}}
\newcommand\ea{\end{align}}
\newcommand\bas{\begin{align*}}
\newcommand\eas{\end{align*}}

\newtheorem{theorem}{Theorem}[section]
\newtheorem{lemma}[theorem]{Lemma}
\newtheorem{proposition}[theorem]{Proposition}
\newtheorem{corollary}[theorem]{Corollary}

\theoremstyle{definition}
\begingroup
\newtheorem{definition}[theorem]{Definition}
\newtheorem{remark}[theorem]{Remark}

\endgroup

\newcommand{\Zz}{\mathbb{Z}}

\newcommand{\eps}{\varepsilon}

\definecolor{WWW}{rgb}{0.0,0.4,0.0}

\definecolor{darkolivegreen}{rgb}{0.33, 0.42, 0.18}

\usepackage{metalogo}
\usepackage{booktabs}
\usepackage{hyperref}
\makeatletter
\hypersetup{%
     pdfpagemode={UseOutlines},
     bookmarksopen,
     pdfstartview={FitH},
     colorlinks,
     linkcolor={blue},
     citecolor={red},
    urlcolor={red}
  }

\begin{document}

\title[Maximal fluctuations of edge-isoperimetric sets in $\mathbb Z^d$]{Maximal fluctuations around the Wulff shape for edge-isoperimetric sets in $\bm{\mathbb Z^d}$: a sharp scaling law}

\author{Edoardo Mainini}
\address[Edoardo Mainini]{Dipartimento di Ingegneria meccanica, energetica, gestionale e dei trasporti, 
  Universit\`a  degli studi di Genova, Via all'Opera Pia, 15 - 16145 Genova Italy.}
\email{mainini@dime.unige.it}
\urladdr{http://www.dime.unige.it/it/users/edoardo-mainini}

\author{Bernd Schmidt}
\address[Bernd Schmidt]{Institut f{\"u}r Mathematik, Universit{\"a}t Augsburg, Universit\"atsstr.\ 14, 86159 Augsburg, Germany}
\email{bernd.schmidt@math.uni-augsburg.de}
\urladdr{http://www.math.uni-augsburg.de/ana/schmidt}

\begin{abstract}
We derive a sharp scaling law for deviations of edge-isoperimetric sets in the lattice $\mathbb Z^d$ from the limiting Wulff shape in arbitrary dimensions. As the number $n$ of elements diverges, we prove that the symmetric difference to the corresponding Wulff set consists of at most $O(n^{(d-1+2^{1-d})/d})$ lattice points and that the exponent $(d-1+2^{1-d})/d$ is optimal. This extends the previously found `$n^{3/4}$ laws' for $d=2,3$ to general dimensions. As a consequence we obtain optimal estimates on the rate of convergence to the limiting Wulff shape as $n$ diverges. 
\end{abstract}



\maketitle

\pagestyle{myheadings}

\section{introduction}
Let $d\in\mathbb N$. For a nonempty subset $C$  of $\Zz^d$, we denote by $\Theta_d(C)$ the {\it edge boundary} of $C$, i.e., 
\[
\Theta_d(C):=\{ (x,y)\in \Zz^d\times\Zz^d \ : \ \textrm{$|x-y|=1$, $x\in C$ and $y\in \Zz^d\setminus C$} \}.
\] 
Its cardinality $\#\Theta_d(C)$  is the {\it edge perimeter} of $C$. Given $n\in\mathbb N$, the $n$-points edge-isoperimetric problem in $\mathbb Z^d$ is the minimization problem
\[
EIP^d(n):=\min \{\#\Theta_d(C): C\subset\mathbb Z^d,\; \#C=n\}.
\]
In the following,  a nonempty set $C$ of $\mathbb Z^d$ is said to be an $EIP^d$ minimizer if the edge perimeter of $C$ is equal to $EIP^d({\# C})$. As a convention, the empty set is assumed to be an $EIP^d$ minimizer as well. 
A solution to the $n$-points edge-isoperimetric problem was given by Bollobas and Leader in \cite{Bollobas}. If two points $x,y$ in a configuration $C \subset \mathbb Z^d$ occupy neighboring lattice sites, i.e.\ $|x - y|=1$, we say there is a {\it bond} connecting these points. The number of bonds $b(C):=\frac12\,\#\{(x,y)\in C\times C: |x-y|=1\}$ satisfies the elementary relation $\#\Theta_d(C)+2b(C)=2d\#C$. This shows that edge-perimeter minimization coincides with number of bonds maximization, as $\# C$ is fixed. 

The edge isoperimetric problem naturally arises within the theory of equilibrium shapes of crystals under a minimal surface energy criterion \cite{BL, DKS, M}. It appears in connection to  low temperature lattice statistics systems such as the Ising model  \cite{ACC,Arous-Cerf,B,C,CK,CM,CP,N}. Regarded as a maximization problem for the number of bonds, it is incurred in the analysis of classical interacting point particle systems with short-range interatomic potentials, where it describes ground states among configurations on a given lattice. In situations where ground states are known to crystallize, $EIP^d$ minimizers are indeed general ground states. Whereas interactions with significant long-range contributions lead to non-trivial boundary layers, see, e.g., \cite{Theil:11,JKST:19}, for specific  {\it sticky-disc} potentials in the plane, crystallization in the triangular lattice has been shown already in \cite{Harborth,Heitmann-Radin80,Radin81}. Yet, convergence of $EIP^2$ minimizers to the hexagonal Wulff shape as the particle number $n$ diverges, the $n^{3/4}$ law for fluctuations at finite $n$ and sharpened estimates with optimal constants have only been obtained rather recently, cf.\ \cite{AuYeung-et-al12,Schmidt,DPS}, respectively. Analogous results for the square lattice and the hexagonal lattice are found in \cite{DPS2,MPS,Mainini-Stefanelli12}, where the different lattice periodicity is induced by the presence of a three-body potential. The emergence of a macroscopic Wulff shape as an effect of the surface tension is a common feature of these models. 

Unlike the classical anisotropic isoperimetric problem in $\mathbb R^d$, which admits  the Wulff shape as the unique solution \cite{DP,F,FonsecaMueller:91,H}, the $n$-points edge isoperimetric problem has many solutions in general. In two dimensions, optimal polyominoes  and lattice animals are  discussed in \cite{BCC,Harary}. Indeed, characterizing isoperimetrically optimal polyominoes and polycubes is a classical problem in discrete mathematics and it  also considered in \cite{CA,  G,EG, NB,VB}. We refer to \cite{Ahlswede, Bezrukov, Bollobas, Harper} for further results in combinatorics and for optimization problems on graphs. 

A peculiar feature of the $EIP^d$ problem is that for infinitely many specific values of $n$ the solution to $EIP^d(n)$ is -- up to translations -- unique (e.g., if $n = \ell^d$ for some $\ell \in \mathbb N$)\footnote{To see this, for each set $C \subset \Zz^d$ with $\# C = n$ let $V_C = \bigcup_{x \in C} (x+[-1/2,1/2]^d)$ with volume $|V_C| = n$ and surface area $\Theta_d(C) = \int_{\partial V_C} \| \nu \|_{L^1}$ ($\nu$ the unit outward normal to $V_C$). As the minimizer of $V \mapsto \int_{\partial^* V} \| \nu \|_{L^1}$ on sets of finite perimeter with volume $n$ is up to translations uniquely given by $[1/2, \ell + 1/2]^d$ (see, e.g., \cite{FonsecaMueller:91,Taylor:75}), every $EIP^d$ minimizer $C$ must satisfy $C = \{1, \ldots, \ell\}^d$ up to translation.}, while for general (infinitely many) $n$ we will see that there are many substantially different minimizers. Our main result Theorem~\ref{thm:main} will show that -- after a suitable translation -- each solution $C$ to $EIP^d(n)$ is close to the cubic Wulff shape $W_n = \{1, \ldots, \lfloor n^{1/d} \rfloor\}^d$ and provide a sharp scaling law for the symmetric distance $C \triangle W_n$ which measures the fluctuations around $W_n$. More precisely, our main result reads as follows. 

\begin{theorem}\label{thm:main} There is a constant $K_d > 0$ which only depends on the dimension $d$ such that 
\begin{itemize} 
\item[(i)] for every $n \in {\mathbb N}$ and each solution $C$ to $EIP^d(n)$ there is a translation vector $a \in \mathbb Z^d$ such that  
$$ \# (C - a) \triangle W_n \le K_d n^{(d-1+2^{1-d})/d}. $$ 

\item[(ii)] This estimate is sharp as for each $\eps > 0$ there are infinitely many $n \in {\mathbb N}$ for which a solution $C$ to $EIP^d(n)$ exists which satisfies the estimate 
$$ \inf_{a \in \mathbb Z^d}\# (C - a) \triangle W_n \ge (K_d - \eps) n^{(d-1+2^{1-d})/d}. $$ 
\end{itemize}
\end{theorem}

We remark that, by way of contrast, the special solutions found  in \cite{Ahlswede, Bollobas} (cf.\ Theorems \ref{theorem:SpecialSolutions} and \ref{unique} below) differ from $W_n$ only on a single surface layer, their symmetric difference thus satisfy the (best possible) estimate of order $O(n^{(d-1)/d})$. 

Still the maximal fluctuations are of lower order than the number $n$ of particles, so that the macroscopic shape of an $EIP^d$ minimizer is close to the Wulff shape as the number of atoms grows. In the setting of Theorem \ref{thm:main} sharp estimates for this convergence can be given by considering the rescaled and translated empirical measure of a sequence $C_n$ of $EIP^d(n)$ minimizers. Rescaling with the edge length $n^{1/d}$, Theorem \ref{thm:main} shows that, for a suitable sequence of translation vectors $a_n$, $\mu_n = \frac{1}{n} \sum_{x \in C_n - a_n} \delta_{x/n^{1/d}}$ converges weakly to the uniform measure on the unit $d$-dimensional cube. Measuring the weak convergence of probability measures in terms of the bounded Lipschitz distance $d_{\rm BL}(\mu,\nu): = \sup_{\varphi \in {\rm Lip}_1} \int_{\mathbb R^d} \varphi \, d(\mu-\nu)$, where ${\rm Lip}_1$ is the space of Lipschitz functions that are bounded by $1$  and have  Lipschitz constant bounded by $1$ as well, Theorem \ref{thm:main} implies 
\[ d_{\rm BL}(\mu_n, \lambda^d|_{[0,1]^d}) \le C n^{(-1+2^{1-d})/d}, \] 
and this estimate is sharp. This convergence is crucial in the context of low temperature crystallization as it provides a theoretical justification for the formation of a deterministic droplet at the macroscopic scale.     

Yet, we also observe that the shape fluctuations at finite $n$ are substantial. Indeed, the non-uniqueness does not solely result from rearrangements of points on the surface. Such differences in `surface particles' would only be of order $O(n^{(d-1)/d})$. Instead, we observe differences of the order $O(n^{(d-1)/d} \cdot n^{2^{1-d}/d})$ which shows that -- in an averaged sense -- microscopic deviations, asymmetries and boundary defects may occur in a whole surface layer of depth $O(n^{2^{1-d}/d})$. (See also the construction in Lemma~\ref{lemma:lb}.) 

Scaling laws for fluctuations around the asymptotic Wulff shape have first been obtained for the planar triangular lattice in \cite{Schmidt}, also cp.\ the announcement in {\cite{AuYeung-et-al12}, and with optimal constants in \cite{DPS}. The square lattice and the hexagonal lattice, including optimal constants, are considered in \cite{MPS, MPS2}, respectively, \cite{DPS2}. More recently, also dimers have been analyzed, cf.\ \cite{FriedrichKreutz:19}. In all the two-dimensional systems an $n^{3/4}$ law was found to sharply describe fluctuations at finite $n$. Very recently, also within the technically much more demanding three-dimensional case a sharp scaling law could be established for the cubic lattice in \cite{MPSS}. Curiously, the same scaling $n^{3/4}$ was found to optimal. The only result in general dimensions appears to be the recent contribution \cite{CL}, which provides another relevant connection between the continuum and the discrete isoperimetric inequality. Indeed, it is shown in \cite{CL} that an estimate from above on the maximal deviation estimate from the Wulff shape in a crystalline system can be obtained through an application of the classical isoperimetric inequality. However, such estimates turn out to be sharp only in dimension $2$, as they provide a higher exponent as compared to the one we find in Theorem \ref{thm:main}. 

To the best of our knowledge, the result of Theorem \ref{thm:main} is the first characterization of the overall shape of edge isoperimetric sets in a higher-dimensional system, providing a sharp scaling law for fluctuations around the perfect cube. Moreover, it closes the analysis for the cubic lattice, clearly recovering the $n^{3/4}$ law in dimension $2$ and $3$. Starting from $d=2$, the sequence of optimal scaling exponents, according to Theorem \ref{thm:main}, turns out to be 
\[
\frac34,\;\;\frac34,\;\;\frac{25}{32},\;\;\frac{13}{16},\;\;\frac{161}{192},\;\;\frac{55}{64},\;\frac{897}{1024},\;\;\frac{683}{768},\;\;\frac{4609}{5120},\;\;\frac{931}{1024},\;\; \frac{22529}{24576},\;\;\ldots, \frac{d-1+2^{1-d}}{d},\;\;\ldots
\]
It is an increasing sequence that converges to $1$ as $d\to+\infty$, consistently with the fact that the number of surface points scales with $n^{(d-1)/d}$ and the total number of points $n$ have the same scaling exponent in the limit. The scaling exponent of the typical averaged width $n^{2^{1-d}/d}$ of surface layers in which boundary defects may occur is found to converge to $0$ geometrically fast as $d \to \infty$.

\subsection*{Plan of the paper} In Section \ref{D} we review the special solutions found in \cite{Ahlswede,Bollobas} and provide some alternative descriptions of such `daisies'. The construction of the lower bound, which is needed to prove Theorem \ref{thm:main}(ii), is given in Section \ref{lower}. The considerably more involved upper bound in Theorem \ref{thm:main}(i) is found in Section \ref{upper}. We close by summarizing our results in the proof of Theorem \ref{thm:main}.

\section{Daisies}\label{D}

We begin by reviewing the special solutions to the edge-perimeter minimization problem that were constructed in \cite{Ahlswede,Bollobas}, see also \cite[Chapter 7]{Harper}. These solutions are obtained by consecutively adding points on hyperplanes neighboring the faces of a cuboid. 

Algebraically, these special solutions are conveniently described in terms of a special order on $\mathbb N^d$. In the following definition we use this notation: for $x=(x_1,\ldots, x_d)\in\mathbb N^d$, we let $\max x:=\max_{i=1,\ldots, d}x_i$, we let $\tilde x=(\tilde x_1,\ldots, \tilde x_d)$, where $\tilde x_i=1$ if $x_i<\max x$ and $\tilde x_i=x_i$ if $x_i=\max x$. Moreover, we let $x_*\in\mathbb N^{d-k}$ be obtained from $x$ by dropping the $k\in\{1,\ldots, d\}$ components of $x$ that are equal to $\max x$.
Finally, we denote by $\prec_{R}$ the right-to-left strict lexicographic order in $\mathbb N^d$, i.e., $x\prec_{ R} y$ if for some $i\in\{1,\ldots, d\}$, there holds 
$$x_j=y_j\;\;\forall j\in\{i+1,\ldots, d\}\quad\mbox{and}\quad x_i<y_i.$$
  
\begin{definition}[Order on $\mathbb N^d$, see \cite{Ahlswede}]\label{ab}
We define a strict and total order relation $\prec$ in $\mathbb N^d$ as follows. For $x=(x_1,\ldots, x_d)\in\mathbb N^d$, $y=(y_1,\ldots, y_d)\in\mathbb N^d$, $x\neq y$, we say that
   $x\prec y$ if one of the following three instances occurs:
\begin{itemize}
\item[1)] $\max x<\max y$
\item[2)] $\max x=\max y$ and $\tilde x\prec_{ R}\tilde y$ 
\item[3)] $\max x=\max y>2$, $\tilde x=\tilde y$, $x_*\prec y_*$
\end{itemize} 
Of course, $x \preceq y$ means $x \prec y$ or $x = y$. 
\end{definition}

We note that in the third instance, since $\tilde x=\tilde y$, the value $\max x=\max y$ is found in $x$ and $y$ exactly at the same entries. If $k\in\{1,\ldots, d-1\}$ is the number of entries that realize such maximum, the relation $x_*\prec y_*$ is defined in the same way but in dimension $d-k$. Therefore the order $\prec$ is defined by induction, and in dimension one $x\prec y \iff x<y$. Given $x\in\mathbb N^d$, $y\in\mathbb N^d$, $x\neq y$, it is easy to check from the above definition that either $x\prec y$ or $y\prec x$, so that $\prec$ is a strict total order in $\mathbb N^d$.
 
\begin{theorem}[Special solutions, see \cite{Ahlswede}]\label{theorem:SpecialSolutions}
For each $n \in \mathbb N$ the string of the first $n$ elements in $\mathbb N^d$ with respect to the order $\prec$ is an $EIP^d$ minimizer. 
\end{theorem}

So in particular one obtains a nested sequence of solutions for any given cardinality. Our first aim is to provide a more geometric characterization of these point sets which in the sequel we refer to as `daisies'. 

\begin{definition}[Perfect daisy]\label{rectdaisy}
Let $k\in\mathbb N$.
 A nonempty set $Q\subset \mathbb Z^k$ is a $k$-dimensional {\it perfect daisy} if it is of the form
$$
Q=\{1,\ldots, p_1^{(k)}\}\times\ldots\times\{1,\ldots, p_k^{(k)}\}
$$
for some natural numbers $p_i^{(k)}$ (called the coefficients of the daisy) such that the sequence $\{1,\ldots, k\}\ni i\mapsto p_i^{(k)}$ is nonincreasing and $p_1^{(k)}-p_k^{(k)}\in \{0,1\}$. 

Tuples $n = (n_1, \ldots, n_k) \in \mathbb N^k$ which are decreasing, i.e., $n_1 \ge \ldots \ge n_k$, and whose oscillation $n_1 - n_k$ is at most $1$ will sometimes be called {\it $DO1$-tuples}. 
We also introduce the {\it value-change position} $s\in\{1,\ldots,k\}$, corresponding to a value change in a $DO1$-tuple $n$, and precisely
\begin{equation}\label{valuechange}
s=s(n_1,\ldots,n_k):=
\left\{\begin{array}{cl}
\min\{j\in\{2,\ldots, k\}: n_j<n_{j-1}\}&\quad \mbox{ if $n_k-n_1=1$},\\
1&\quad\mbox{ if $n_k-n_1=0$}.
\end{array}\right.
\end{equation}
\end{definition}

\begin{definition}[Daisy]
\label{dai} Let $d\in\mathbb N$.
 A nonempty set $Q\subset \mathbb Z^d$ is a $d$-dimensional {\it daisy} if for some $h\in\{0,\ldots d-1\}$ it is of the form
$$
Q=Q^{(d)}\cup Q^{(d-1)}\cup\ldots\cup Q^{(d-h)}, \quad\mbox{where}
$$
1) $Q^{(d)}$ is a $d$-dimensional perfect daisy  (Definition {\rm \ref{rectdaisy}}), with coefficients ${q}_i^{(d)}$, $i\in\{1,\ldots, d\}$.

\noindent 2) A sequence $(s_k)\subset\{1,\ldots,d\}$  and the nonempty sets $Q^{(d-k)}$ are defined recursively for $k=1,\ldots, h$ as follows:  
$Q^{(d-k)}=A_1^{(d-k)}\times\ldots\times A_d^{(d-k)}$, $S_{d,0}:=\{1,\ldots, d\}$,
\[
A_j^{(d-k)}:=\left\{\begin{array}{ccc}\{1,\ldots, q_j^{(d-k)}\}&\quad\mbox{if $j\in S_{d,k}:=\{1,\ldots, d\}\setminus\{s_1,\ldots,s_{k}\}$}\\
\{q_j^{(d-r_{j,k})}+1\}&\quad \mbox{if $j\in\{s_1,\ldots, s_k\}$},
\end{array}\right.
\] 
where
$r_{j,k}:=\max\{n\in\{0,\ldots, k-1\}: j\in S_{d,n}\},$ and
\[
s_k:=\left\{
\begin{array}{lll}
\min S_{d,k-1}&\;\; \mbox{ if $q_i^{(d-k+1)}=q_j^{(d-k+1)}$   for any $i,j\in S_{d,k-1}$}
\\
\min\{j\in  S_{d,k-1}:  q_{\pi(j)}^{(d-k+1)}>q_j^{(d-k+1)}\}&\;\;\mbox{ otherwise},
\end{array}
\right.
\]
where, for $j\in S_{d,k}$ such that $j>\min S_{d,k}$, the notation is $\pi(j):=\max\{i\in S_{d,k}:i<j\}$. 
If $Q^{(1)}\neq \emptyset$ we also conventionally denote by $s_d$ the unique element of $S_{d,d-1}$. 

\noindent 3) For any $k=\{1,\ldots,h\}$, the  natural numbers $q_j^{(d-k)}$ are defined for $j\in S_{d,k}$ and satisfy
\begin{itemize}
\item[3.1)] for $i,j\in S_{d,k}$, there holds $i<j\Rightarrow q_i^{(d-k)}\ge q_j^{(d-k)}$,
\item[3.2)] for $J_1:=\min S_{d,k}$ and $J_2=\max S_{d,k}$, 
there holds $q_{J_1}^{(d-k)}-q_{J_2}^{(d-k)}\in\{0,1\}$, 
\item[3.3)] for all $i\in S_{d,k}$, there holds $q_i^{(d-k)}\le q_i^{(d-k+1)}$,
\item[3.4)] there exists $i\in S_{d,k}$ such that there holds    $q_i^{(d-k)}< q_i^{(d-k+1)}$.
\end{itemize}
\end{definition}

\begin{remark}\rm 
The sets $Q^{(d-k)}$, $k\in\{1,\ldots h\}$, from Definition \ref{dai} are all nonempty. However, we shall often denote a $d$-dimensional daisy $Q$ as $Q^{(d)}\cup\ldots\cup Q^{(1)}$ even if $h<d-1$ in Definition \ref{dai}.
In such case, it is understood that $Q=Q^{(d)}\cup\ldots\cup Q^{(d-h)}$, where $Q^{(d-k)}\neq\emptyset$ if $k\in\{1,\ldots h\}$ and $Q^{(d-k)}=\emptyset $ if $k\in\{h+1,\ldots, d-1\}$.
\end{remark}

The description in Definition \ref{dai} is rather involved mainly due to the fact that the precise description of the position of an individual constituent $Q^{(m)}$ which is merely an (isometric) copy of a perfect $m$-dimensional daisy is quite complicated. We therefore provide an alternative description in terms of a collection of perfect daisies with a compatibility condition.  

\begin{definition}[Larger sequences]\label{larger} Let $(a_1,\ldots, a_n) \in \mathbb N^n$ and $(b_1,\ldots, b_{n+1}) \in \mathbb N^{n+1}$ be $DO1$-tuples. We say that $(b_1,\ldots, b_{n+1})$ is {\it larger} than $(a_1,\ldots, a_n)$ and write $(a_1,\ldots, a_n)\sqsubset(b_1,\ldots,b_{n+1})$ if $a_i\le b_{f(i)}$ for any $i\in\{1,\ldots, n\}$ and strict inequality holds for at least one of the indices $i=1,\ldots, n$. Here, $f$ is the increasing bijection from $\{1,\ldots, n\}$ onto $\{1,\ldots, n+1\}\setminus\{s\}$, where $s\in\{1,\ldots, n+1\}$ is the position corresponding to a value change in the sequence $(b_1,\ldots,b_{n+1})$, which is defined as in \eqref{valuechange}. 
\end{definition}

\begin{proposition}\label{characterization}
 A  $d$-dimensional daisy $Q=Q^{(d)}\cup\ldots\cup Q^{(d-h)}$ identifies with a collection of $(d-k)$-dimensional perfect daisies (according to  {\rm Definition \ref{rectdaisy}}), still denoted by $Q^{(d-k)}$, $k= 0,\ldots,h$, with coefficients $p_i^{(d-k)}$, $i=1,\ldots,d-k$, such that 
$(p_1^{(d-k)},\ldots, p_{d-k}^{(d-k)})\sqsubset(p_1^{(d-k+1)},\ldots, p_{d-k+1}^{(d-k+1)})$ for any $k\in\{1,\ldots, h\}$ in the sense of {\rm Definition \ref{larger}}. 
\end{proposition}

\begin{proof} 
Given a daisy from Definition \ref{dai}, we introduce the increasing bijection $b:\{1,\ldots, d-k\}\to S_{d,k}$ and coefficients $p_i^{(d-k)}:=q_{b(i)}^{(d-k)}$, $i=1,\ldots, d-k$, so that
\[
\prod_{i\in S_{d,k}}\{1,\ldots, q_i^{(d-k)}\}=\prod_{i=1}^{d-k} \{1,\ldots,p_i^{(d-k)}\}.
\] 
For any $k\in\{0,\ldots, h\}$, the sequence $\{1,\ldots, d-k\}\ni i\mapsto p_i^{(d-k)}$ is $DO1$, thanks to properties 3.1) and 3.2) of Definition \ref{dai}. In other words, any layer $Q^{(d-k)}$ can be identified with a $(d-k)$-dimensional perfect daisy with coefficients $p_i^{(d-k)}$, $i=1,\ldots,d-k$, according to Definition \ref{rectdaisy}, by dropping from any point $z=(z_1,\ldots, z_d)\in Q^{(d-k)}$ all the components $z_i$ such that $i\notin S_{d,k}$. 
Moreover, by properties 3.3) and 3.4) of Definition \ref{dai} we infer that $(p_i^{(d-k)},\ldots, p_{d-k}^{(d-k)})\sqsubset(p_1^{(d-k+1)},\ldots, p_{d-k+1}^{(d-k+1)})$, for any $k\in\{1,\ldots, h\}$, in the sense of {\rm Definition \ref{larger}}. 

On the other hand, given $DO1$-sequences $\{1,\ldots, k\}\ni i\mapsto p_i^{(d-k)}$ for $k\in\{0,\ldots,h\}$, suppose that $(p_i^{(d-k)},\ldots, p_{d-k}^{(d-k)})\sqsubset(p_1^{(d-k+1)},\ldots, p_{d-k+1}^{(d-k+1)})$ for any $k\in\{1,\ldots, h\}$.  Then, the numbers $s_j$ from Definition \ref{dai} are uniquely identified in terms of the value-change positions of these sequences. Indeed, we define $Q^{(d)}$ as the perfect $d$-dimensional daisy with coefficients $\{p_1^{(d)},\ldots, p_d^{(d)}\}$, then we define  $s_1$ as the value-change position for the sequence $(p_1^{(d)},\ldots, p_d^{(d)})$ according to formula \eqref{valuechange}, $S_{d,1}:=\{1,\ldots,d\}\setminus\{s_1\}$ and we define for $i\in S_{d,1}$ the numbers $q^{(d-1)}_i:=p^{(d-1)}_{g_1(i)}$, where $g_1(i)$ is the increasing bijection of $S_{d,1}$ onto $\{1,\ldots,d-1\}$.
Then we define $s_2$ from $S_{d,1}$ and from the sequence, $(q_i^{(d-1)})_{i\in S_{d,1}}$ as done in Definition \ref{dai}. Therefore, we recursively define, for $k=2,\ldots, h$, the numbers $q_i^{(d-k)}:=p^{(d-k)}_{g_k(i)}$, where $g_k(i)$ is the increasing bijection of $S_{d,k}$ onto $\{1,\ldots, d-k\}$, and then $s_{k+1}$ from $S_{d,k} = \{1, \ldots, d\}\setminus\{s_1,\ldots, s_k\}$ and the coefficients $q_i^{(d-k)}$ as done in Definition \ref{dai}. The relation $\sqsubset$ between sequences $p_i^{(k)}$ ensures that properties 3.3) and 3.4) of Definition \ref{dai} are satisfied. 
\end{proof}

\begin{remark}\label{rmk:daisy-constituents}\rm
A $d$-dimensional daisy $Q=Q^{(d)}\cup\ldots\cup Q^{(1)}$ can be characterized either by the coefficients $q_i^{(k)}$ from Definition \ref{dai} or by the coefficients $p_i^{(k)}$ from Proposition \ref{characterization}. In the sequel we will also refer to a subset of $\mathbb Z^d$ which is an isometric copy of an $m$-dimensional daisy ($m \le d$) simply as a daisy (as, e.g., in Proposition \ref{daisysection} and Corollary \ref{coro1} below). In particular, the constituents $Q^{(m)}$ of $Q$ are $m$-dimensional daisies. 
\end{remark}

In order to see that daisies are in fact the solutions found in Theorem \ref{theorem:SpecialSolutions} we note that, in view of Definition \ref{dai} and Proposition \ref{characterization}, daisies can also be characterized by matrices. To this end, we let $\mathcal{A}$ be the set of $(h+1) \times d$ matrices $ A = (a_{i,j})_{1 \le i \le h+1 \atop 1 \le j \le d}$ with $h \le d-1$ whose entries consist of dots and numbers in the following way. The first line $(a_{1,1}, \ldots, a_{1,d})$ is a $DO1$-tuple. The second line has a dot at the value change position $s_1 = s(a_{1,1}, \ldots, a_{1,d})$ of the first line, i.e., $a_{2,s_1} = \cdot$, and $(a_{2,1}, \ldots, a_{2,s_1-1},a_{2,s_1+1},\ldots,a_{2,d})$ is $DO1$ with $(a_{2,1}, \ldots, a_{2,s_1-1},a_{2,s_1+1},\ldots,a_{2,d}) \sqsubset (a_{1,1}, \ldots, a_{1,d})$. In general, the $i$-th line consists of $i-1$ dots at the positions $s_1, \ldots, s_{i-1}$, where $s_{k}$ is the value change position of the sequence of numbers in the $k$-th line, $k = 1, \ldots, i-1$, and the tuple of numbers that is obtained by omitting these dots is a $(d-i+1)$-dimensional $DO1$-tuple which is smaller (wrt $\sqsubset$) than the sequence of numbers in the previous line. 

Note that the set of daisies is in one-to-one correspondence with the set $\mathcal{A}$: If we denote the sequence of numbers in the $i$-th line of $A \in \mathcal{A}$ by $(p^{(d-i+1)}_1, \ldots, p^{(d-i+1)}_{d-i+1})$, $A$ corresponds to the daisy $Q = Q^{(d)} \cup \ldots \cup Q^{(d-h)}$ with $Q^{(d-i+1)} = \{1, \ldots, p_1^{(d-i+1)}\} \times \ldots \times \{1, \ldots, p_{d-i+1}^{(d-i+1)}\}$, $i = 1, \ldots, h+1$, and, conversely, each daisy arises in such a way, see Proposition \ref{characterization}. With respect to the geometric position of the individual perfect daisy $Q^{(d-i+1)}$, as detailed in Definition \ref{dai}, we note that the numbers within the $i$-line are also the $q_i^{(d-k+1)}$ coefficients and dots occupy the positions $s_j$ for $j\in \{1,\ldots, {i-1}\}$. A number $a$ in the matrix corresponds to the factor $\{1,\ldots, a\}$, and any dot in a column corresponds to the factor $\{a+1\}$, where $a$ is the first number that is found going up in such column. Finally we observe that the cardinality of the daisy is just the line by line sum of the product of all the numbers in each line. 

\noindent {\em Example.} Two $5$-dimensional examples of $Q=\cup_{k=1}^5 Q^{(k)}$ are 
\begin{equation*}
\begin{pmatrix} 5 & 5 & 4 & 4 & 4\\
                             4 & 3 & \cdot  & 3 & 3\\
                             3 &  \cdot &   \cdot & 3 & 2\\
                             2 &  \cdot  &   \cdot & 2 &  \cdot \\
                              \cdot  &   \cdot &   \cdot & 1 &  \cdot 
\end{pmatrix}\qquad\qquad
\begin{pmatrix} 7 & 7 & 7 & 7 & 7\\
                             \cdot & 4 & 3  & 3 & 3\\
                             \cdot & 3 & \cdot &   3 & 2            
\end{pmatrix}
\end{equation*}
In the second example, $Q^{(4)}=Q^{(5)}=\emptyset$.

\noindent {\em Example.} Two-dimensional daisies are subsets of $\mathbb Z^2$ of the form
\begin{align}\label{eq:two-d-daisy} 
  D^{(2)}_{a,b,c}
  :=\left\{\begin{array}{ccc}(\{1,\ldots, a\}\times\{1,\ldots, b\})\cup(\{a+1\}\times\{1,\ldots, c\})&\quad\mbox{ if $b=a$,}\\
  (\{1,\ldots, a\}\times\{1,\ldots, b\})\cup(\{1,\ldots, c\}\times\{b+1\})&\quad\mbox{ if $b+1=a$,}
\end{array}\right.
\end{align}
for given $b\in\mathbb N$, $a\in\{b,b+1\}$ and $c\in\{0,\ldots,a-1\}$, where it is understood that $\{1,\ldots c\}=\emptyset$ in case $c=0$.

Indeed, we have $D^{(2)}_{a,b,c}=Q^{(2)}\cup Q^{(1)}$, with $q_1^{(2)}=a$ and $q^{(2)}_2=b$ representing the coefficients of the perfect daisy $Q^{(2)}$. Moreover, we have $S_{2,0}=\{1,2\}$, $S_{2,1}=S_{2,0}\setminus \{s_1\}$, where
\[s_1=\left\{
\begin{array}{ll}
1&\quad\mbox{if $b=a$}\\
2&\quad\mbox{if $b+1=a$},
\end{array}\right.
\]
and $Q^{(1)}=A^{(1)}_1\times A^{(1)}_2$, where
\[
A_1^{(1)}=\left\{\begin{array}{ll}\{1,\ldots, c\}&\quad \mbox{if $s_1=2$},\\
\{1+a\}&\quad\mbox{if $s_1=1$},
\end{array}\right.\qquad \quad
A_2^{(1)}=\left\{\begin{array}{ll}\{1,\ldots, c\}&\quad\mbox{if $s_1=1$},\\
\{1+b\}&\quad\mbox{if $s_1=2$}.
\end{array}\right.
\]
Or simply in matrix form 
\[\begin{pmatrix} a & b\\
                             \cdot & c \\
\end{pmatrix} \quad\mbox{ if $a=b$},\qquad\quad  
\begin{pmatrix} a & b\\
                             c & \cdot \\
\end{pmatrix} \quad\mbox{if $a=b+1$},  
\]
reduced to $(a\;\;b)$ if $c=0$ (i.e. $Q^{(2)}=\emptyset$).

\begin{theorem}[Daisies are unique and $EIP^d$ minimizers]\label{unique}
For $n,d\in\mathbb N$, there exists a unique $d$-dimensional daisy $Q$ such that $\#Q=n$. Moreover, it coincides with the string of the first $n$ elements in $\mathbb N^d$ with respect to the order $\prec$. In particular, $Q$ is an $EIP^d$ minimizer. 
\end{theorem}

\begin{proof}
In view of Theorem \ref{theorem:SpecialSolutions} and our identification of daisies with matrices in $\mathcal{A}$, it suffices to show that there is a bijective mapping $\Phi : \mathcal{A} \to \mathbb N^d$ such that the daisy corresponding to $A \in \mathcal{A}$ is given by $\{ m \in \mathbb N^d : m \preceq \Phi(A) \}$. 

To define such $\Phi$ consider the last row $(a_{h+1,1}, \ldots, a_{h+1,d})$ of the daisy matrix $A = (a_{ij})_{1 \le i \le h+1 \atop 1 \le j \le d} \in \mathcal{A}$ and replace each dot $a_{h+1,j}$ with $a_{i,j}+1$ if $a_{i,j}$ is the first number that is found going up in column $j$. We define $n = \Phi(A) \in \mathbb N^d$ to be the $d$-tuple thus obtained. 

Conversely, suppose a tuple $n = (n_1, \ldots, n_d) \in \mathbb N^d$ is given. We define an $A = \Psi(n) \in \mathcal{A}$ by induction on the lines of $A$. If $n$ is a $DO1$-sequence, we stop and set $A = n$ (a perfect daisy). If $n$ is not a $DO1$-sequence, we consider the rightmost entry $n_j$ for which the maximum is attained, i.e., $n_j = \max\{n_1, \ldots, n_d\} > n_{j+1}, \ldots, n_{d}$ and let $a_{1,1} = \ldots = a_{1,j-1} = n_j$,  $a_{1,j} = \ldots = a_{1,d} = n_j-1$. (This is the largest $DO1$-sequence which is dominated by $n$.) We also fill the rest of the $j$-th column with dots. If $n' = (n_1, \ldots, n_{j-1}, n_{j+1}, \ldots, n_d)$ is a $DO1$-sequence, we set $(a_{21}, \ldots, a_{2,j-1}, a_{2,j+1}, \ldots a_{2,d}) = n'$ and stop (obtaining a daisy with $h = 1$). If not, we continue this procedure until a $DO1$-sequence is reached. Note that our choice of the rightmost maximal entry as the value-change position for the constructed $DO1$-sequence guarantees that indeed the sequence of numbers in a line of $A$ is always larger than the sequence of numbers in the next line of $A$. 

The assertion of Theorem \ref{unique} now follows from the following two observations: $\Phi$ and $\Psi$ are inverse to each other and the daisy decried by an $A \in \mathcal{A}$ is given by $\{ m \in \mathbb N^d : m \preceq \Phi(A) \}$. 

In order to see that $\Psi \circ \Phi = \mathrm{id}$ consider $A \in \mathcal{A}$ and set $n = \Phi(A)$. We observe that since the $DO1$-sequences of numbers within the lines of $A$ are ordered wrt $\sqsubset$, the index $s_1$ of the rightmost maximum of $n$ is the value change position of the first line and its value $n_{s_1}$ is given by $a_{1,s_1}+1$. This shows that the first line of $\Psi \circ \Phi(A)$ is indeed $(a_{1,1}, \ldots, a_{1,d})$. Now deleting the first line and $s_1$-th column, the same argument for the remaining part shows that the second line is reproduced correctly as well. Continuing in this way, wee indeed get that $\Psi \circ \Phi = \mathrm{id}$. 

To prove that also $\Phi \circ \Psi = \mathrm{id}$ we start with $n \in \mathbb N^d$ and set $A = \Psi(n)$. If $n$ is a $DO1$-sequence, clearly $\Phi(A) = n$. If not, then by $j$ denoting the largest index for which $n_j = \max\{n_1, \ldots, n_d\}$, we have $a_{1j} = n_j-1$ and $a_{ij} = \cdot$ if $j \ge 2$. By definition of $\Phi$ this gives $(\Phi(A))_j = n_j$. If $n' = (n_1, \ldots, n_{j-1}, n_{j+1}, \ldots, n_d)$ is a $DO1$-sequence, we also have set $(a_{21}, \ldots, a_{2,j-1}, a_{2,j+1}, \ldots a_{2,d}) = n'$ and so $\Phi(A) = n$. If not, we continue repeating the above step to finally obtain that indeed $\Phi(A) = n$.

Now suppose $A \in \mathcal{A}$ representing a daisy $Q = Q^{(d)} \cup \ldots \cup Q^{(h)}$ is given. We define $\tilde{A} = (\tilde{a}_{i,j})_{1 \le i \le h+1 \atop 1 \le j \le d}$ $(h \le d-1)$ by replacing each dot in $A$ with the coordinate it represents: For each column $j$, if $a_{1,j}, \ldots, a_{i,j} \neq \cdot$ and $a_{i+1,j} = \ldots = a_{h,j} = \cdot$, then $\tilde{a}_{i+1,j} = \ldots = \tilde{a}_{h,j} = a_{i,j} + 1$ while $\tilde{a}_{k,j} = a_{k,j}$ for $1 \le k \le i$. Recall that here $j$ is a value-change position of the $i$-line. So in fact the lines of $\tilde{A}$ are increasing with respect to $\prec$: $(\tilde{a}_{11}, \ldots, \tilde{a}_{1d}) \prec \ldots \prec (\tilde{a}_{h1}, \ldots, \tilde{a}_{hd})$. Also, by construction each perfect daisy $Q^{(k)}$ consists of precisely those points $m \in \mathbb N^d$ which satisfy $(\tilde{a}_{k-1,1}, \ldots, \tilde{a}_{k-1,d}) \prec m \preceq (\tilde{a}_{k,1}, \ldots, \tilde{a}_{k,d})$. Thus, $Q = \{ m \in \mathbb N^d : m \preceq \Phi(A) \}$.
\end{proof}

\begin{remark}[Explicit construction of daisies]\label{rmk:daisy-construction}\rm 
Explicitly, one finds the coefficients $p_i^{(d-k)}$, $i=1,\ldots,d-k$, $k=0,\ldots, h$ of a daisy $Q=Q^{(d)}\cup\ldots\cup Q^{(d-h)}$ of given cardinality $n$ inductively: $(p_1^{(d)},\ldots, p_{d}^{(d)})$ is the largest $DO1$-tuple wrt $\prec$ of length $d$ such that $p_1^{(d)} \cdot \ldots \cdot p_{d}^{(d)} \le n$ and, for $k \ge 1$, $(p_1^{(d-k)},\ldots, p_{d-k}^{(d-k)})$ is the largest $DO1$-tuple wrt $\prec$ of length $d-k$ such that $p_1^{(d-k)} \cdot \ldots \cdot p_{d-k}^{(d-k)} \le n - \# Q^{(d)} - \ldots - \# Q^{(d-k-1)}$ as long as this number is not zero. If it is zero for the first time, let $h = k+1$. Note that indeed 
\[ (p_1^{(d-k)},\ldots, p_{d-k}^{(d-k)})\sqsubset(p_1^{(d-k+1)},\ldots, p_{d-k+1}^{(d-k+1)}) \] 
for any $k\in\{1,\ldots, h\}$ since by construction, if $s(p_1^{(d-k+1)},\ldots, p_{d-k+1}^{(d-k+1)}) = s$ and $p_{d-k+1}^{(d-k+1)} =: p$, then 
\[ p_1^{(d-k)} \cdot \ldots \cdot  p_{d-k}^{(d-k)}) 
   < (p+1)^s p^{d-s} - (p+1)^{s-1} p^{d-s+1} 
   = (p+1)^{s-1} p^{d-s} \] 
and so $(p_1^{(d-k)} \cdot \ldots \cdot  p_{d-k}^{(d-k)}) \prec (p_1^{(d-k+1)},\ldots, p_{s-1}^{(d-k+1)}, p_{s+1}^{(d-k+1)} \ldots, p_{d-k+1}^{(d-k+1)})$. 
\end{remark} 

We conclude this section with a property of faces and sections of daisies. There is a similar result for general $EIP^d$ minimizers, see Corollary \ref{coro1}.

\begin{definition}[Sections]\label{sect}
Let $C\subset\mathbb Z^d$ be a nonempty set. For $s\in\{1,\ldots, d\}$ and $k\in\mathbb Z$ we define the $(d-1)$-dimensional {\it section} $S_{s,k}(C) := \{x\in C:\mathbf e_s\cdot x=k\}$ of $C$. 
\end{definition}

\begin{definition}[Faces]\label{faces}
If $\emptyset \neq C\subset\mathbb Z^d$, any nonempty $(d-1)$-dimensional section $S_{s,k}(C)$ for which $S_{s,k+1}(C) = \emptyset$ or $S_{s,k-1}(C) = \emptyset$ is called a {\it (lateral) face} of $C$ (with normal $\mathbf e_s$). If $P$ is a perfect $d$-dimensional daisy and $m \in \{0, \ldots, d-2\}$, we also define an {\it $m$-dimensional face} of $P$ to be any (nonempty) subset of the form $L_1 \cap \ldots \cap L_{d-m}$, where $L_i$ is a lateral face of $P$ with normal $\mathbf e_{s_i}$ and $1 \le s_1 < \ldots < s_{d-m} \le d$.  
\end{definition}

\begin{proposition}\label{daisysection}
Each $(d-1)$-dimensional section of a $d$-dimensional daisy is a $(d-1)$-dimensional daisy. 
\end{proposition}

\begin{proof}
Let $Q$ be a $d$-dimensional daisy and wlog assume that that $S_{s,k}(Q)\neq\emptyset$. Let $P : S_{s,k}(\mathbb N^d) \to \mathbb N^{d-1}$ be the bijective mapping $P(z_1, \ldots, z_{s-1}, k, z_{s+1}, \ldots z_d) = (z_1, \ldots, z_{s-1}, z_{s+1}, \ldots z_d)$.  We identify $S_{s,k}(Q)$ with $\mathcal S := P(S_{s,k}(Q))$. Now observe that each point of $\mathcal S$ can be written as $P(v)$ for some $v\in S_{s,k}(Q)\subseteq Q$ and each point in $\mathbb N^{d-1}\setminus \mathcal S$ can be written as $P(w)$ for some  $w\in S_{s,k}(\mathbb N^{d}\setminus Q)\subseteq \mathbb N^d\setminus Q$. Therefore, we have $v\prec w$ by Theorem \ref{unique}. Since $w_s = k = v_s$ this also gives $ \mathcal S \ni P(v)\prec P(w)\notin \mathcal S $. We have thus proven that for any $x\in\mathcal S$ and any $y\notin\mathcal S$, there holds, $x\prec y$. This shows that $\mathcal S$ is the string of the first $\#\mathcal S$ points of $\mathbb N^{d-1}$ with respect to the order relation $\prec$. By Theorem \ref{unique}, $\mathcal S$ is a daisy.
\end{proof}

\section{Lower bound}\label{lower}

\begin{definition}[Scaling parameter]\label{parameter}
For $\ell\in\mathbb N$, $d\in\mathbb N$ we define $h_{\ell,d}:=\ell^{\,2^{1-d}}$.
\end{definition}

The next statement makes use of the notation of Definition \ref{sect}. It extends some rearrangement procedures that have already been introduced in \cite{MPS, MPSS}, whose main property is the monotonicity of the edge perimeter. 

\begin{proposition}[Decreasing rearrangement]\label{rear}
Let $C\in\mathbb Z^d$ be a bounded nonempty set. Let $s\in\{1,\ldots, d\}$ and $k\in\mathbb Z$. Let $K_s:=\{k_1,\ldots, k_n\}$ denote the  finite strictly increasing sequence of  integers such that $S_{s,k}(C)\neq\emptyset\iff k\in K_s$. Let $\sigma:\{1,\ldots,n\}\to K_s$ be a bijection such that  $\#S_{s,\sigma(i)}(C)\ge \#S_{s,\sigma(j)}(C)$ for any $1\le i\le j\le n$.
Let $D^{(d-1)}_{s,k}$ be the $(d-1)$-dimensional daisy with the same cardinality as $S_{s,k}(C)$. Finally, let $C_s\subset \mathbb Z^d$ denote the decreasing rearrangement of $C$ in the $\mathbf e_s$ direction, i.e., the unique configuration whose nonempty sections orthogonal to $\mathbf e_s$ are given by $P S_{s,k}(C_s)=D^{(d-1)}_{s,\sigma(k)}$, $k=1,\ldots, n$, where $P(z_1, \ldots, z_d) = (z_1, \ldots, z_{s-1}, z_{s+1}, \ldots z_d)$. Then $\#\Theta_d(C_s)\le \#\Theta_d(C)$.
\end{proposition}
\begin{proof}
For any $k\in K_s$, we look at  $(d-1)$-dimensional configurations and we have $b(D^{(d-1)}_{s,k})\ge b(S_{s,k}(C))$, since daisies minimize the edge perimeter and maximize the number of bonds. This shows that the total number of bonds in directions that are orthogonal to $\mathbf e_s$ does not increase after the rearrangement. If $n=1$, the proof is concluded. Suppose instead that $n>1$, and we are left to check the number $b_s(\cdot)$ of bonds  in the direction of $\mathbf e_s$. For $k\in K_s$ we use  the shorthand $f(k):=\#S_{s,k}(C)=\# D^{(d-1)}_{s,k}$. Moreover, we define $I \in\{ 1, \ldots, n \}$ such that $k_I = \sigma(1)$ so that $f(k_I)\ge f(k_i)$ for any $i\in\{1,\ldots, n\}$.
 By counting the bonds in the $\mathbf e_s$ direction
as sum of bonds between couples of consecutive sections, 
 we have
\[\begin{aligned}
b_s(C)&\le \sum_{i=2}^n \min\{f(k_{i-1}), f(k_i)\}\le 
 \sum_{i\in\{1,\ldots, n\}\setminus\{I\}} f(k_i)=\sum_{i=2}^n f(k_{\sigma(i)})=b_s(C_s),
\end{aligned}
\]
where the second inequality is obtained by using $\min\{f(k_{i-1}), f(k_i)\}\le f(k_{i-1})$ for $i\in\{2,\ldots, I\}$ (only in case $I>1$) and $\min\{f(k_{i-1}), f(k_i)\}\le f(k_i)$ if $i\in\{I+1,\ldots, n\}$.
The proof is concluded.
\end{proof}

Arguing by contradiction we deduce the following result (whose converse is false as seen already in dimension $2$ by taking 
a configuration such as $\{(1,1),(1,2),\ldots,(1,n)\}$, $n\in\mathbb N, n\ge 4$).
\begin{corollary}\label{coro1}
Let $C$ be an $EIP^d$ minimizer. Then each $(d-1)$-dimensional section is an $EIP^{d-1}$ minimizer.
\end{corollary}

\begin{proof}
If $S_{s,k}(C)$ were not an $EIP^{d-1}$ minimizer, then $\#\Theta_{d-1}(S_{s,k}(C)) > \#\Theta_{d-1}(D_{s,k}^{(d-1)})$ and the above proof shows $\#\Theta_{d}(C_s) < \#\Theta_{d}(C)$. 
\end{proof}

\begin{lemma}\label{easy2}
Let $\ell\in\mathbb N$.
Let $p\in\mathbb N$ be such that $p<\ell$. Suppose that
\[
M:=\{1,\ldots, \ell-p\}\times\{1,\ldots,\ell\}^{d-2}\times\{1,\ldots, \ell+p\}
\]
is an $EIP^d$ minimizer. Then
\[
Q:=\{1,\ldots, \ell-p\}\times\{1,\ldots,\ell\}^{d-1}
\]
is an $EIP^d$ minimizer as well.
\end{lemma}

\begin{proof}
We observe that $M=Q\cup T$, where $T:=\{1,\ldots,\ell-p\}\times\{1,\ldots,\ell\}^{d-2}\times\{\ell+1,\ldots,\ell+p\}.$
The number of bonds connecting these two blocks is $(\ell-p)\ell^{d-2}$. 

We take the decreasing rearrangement (see Proposition \ref{rear}) of $M$ in the direction of $\mathbf e_d$. We get a configuration $\overline{M}$ whose sections $S_{d,k}(\overline M)$ are nonempty for $k=1,\ldots, \ell+p$ so that $\overline M=\bigcup_{k=1}^{\ell+p} S_{d,k}(\overline M)$. By considering $S_{d,k}(\mathbb Z^d)$ as a copy of $\mathbb Z^{d-1}$, each of such sections identifies with the $(d-1)$-dimensional daisy of cardinality  $(\ell-p)\ell^{d-2}$. Since $M$ is an $EIP^d$ minimizer, then $\overline{M}$ is an $EIP^d$ minimizer as well by Proposition \ref{rear}, and it is itself a union of  two blocks  $\overline Q$ and $\overline T$, where
\[
\overline Q:= \bigcup_{k=1}^{\ell} S_{d,k}(\overline M),\qquad \overline T:= \bigcup_{k=\ell+1}^{\ell+p} S_{d,k}(\overline M),
\]
with $\# Q=\#\overline Q$, $b(Q)=b(\overline Q)$, $\# T=\#\overline T$, $b(T)=b(\overline T)$, and 
\begin{equation}\label{qt}
b(\overline M)=b(\overline T)+b(\overline Q) + (\ell-p)\ell^{d-2}
\end{equation}

Now, assuming  that $Q$ is not an $EIP^d$ minimizer, we shall prove that $\overline M$ is not an $EIP^d$ minimizer either, thus reaching a contradiction and concluding the proof. Indeed, if $Q$ is not an $EIP^d$ minimizer, we consider the daisy $D$ with the same cardinality so that 
\begin{equation}\label{DQ} 
   (\ell-p)\ell^{d-1}=\#D=\#Q=\#{\overline Q}
\end{equation} 
and 
\begin{equation}\label{bDQ}
   b(D) > b(Q)=b(\overline Q).
\end{equation} 
$D$ is of course contained in the daisy $\{1,\ldots, \ell\}^d$ whose cardinality is larger, since daisies are ordered by cardinality, see Theorem \ref{unique}. In particular, by looking at its sections in the direction of $\mathbf e_d$, we see that for some $1 \leq h \leq \ell$ we have $S_{d,k}(D)\neq\emptyset$ if and only if $k\in\{1,\ldots, h\}$. Moreover, each nonempty section $S_{d,k}(D)$  identifies with $EIP^{d-1}$ minimizers (see Corollary \ref{coro1}). We claim that  $S_{d,1}(D)$ identifies with a $(d-1)$-dimensional daisy and  $\#S_{d,1}(D)\ge(\ell-p)\ell^{d-2}$. 
Indeed, the fact that $S_{d,1}(D)$ is a $(d-1)$-dimensional  daisy comes from Proposition \ref{daisysection}.
Moreover, from Definition \ref{dai} it is possible to see that $\#S_{d,i}(D)\ge \#S_{d,j}(D)$ if $1\le i\le j$: this fact can be alternatively deduced from  Theorem \ref{unique}, since Definition \ref{ab} readily implies that if $x=(x_1,\ldots, x_{d})\in D$,  then  $(x_1,\ldots, x_{d-1}, y_d)\prec x$ for any $y_d\in\{1,\ldots, x_{d-1}\}$.
 Therefore $\#D\le h\, \#S_{d,1}(D)$, so that if $\#S_{d,1}(D)<(\ell-p)\ell^{d-2}$ were true it would lead to $\#D<(\ell-p)\ell^{d-1}$, which is against \eqref{DQ}. The claim is proved.

We take
a rigid motion of $\overline T$ in the direction of $\mathbf e_d$, i.e., we introduce $T^*:= \overline T-(\ell+p)\mathbf e_d$, so that
\[
T^*=\bigcup_{k=1-p}^{0} S_{d,k}( T^*)
\]
Then we let $ M^*:=D\cup T^*$. The cardinality of $ M^*$ is that of $\overline M$, since \eqref{DQ} holds and since obviously $\# T^*=\#\overline T$.  
Similarly, $b(T^*)=b(\overline T)$. Most importantly, $$\mathrm{dist}(D, T^*) =\mathrm{dist}(S_{d,1}(D),S_{d,0}( T^*))=1$$ and the number of bonds connecting $D$ and $T^*$ is equal to $\# S_{d,0}( T^*)$: indeed, each point of the form $S_{d,0}( T^*)+\mathbf e_d$ belongs to $S_{d,1}(D)$, because we have already proven that $S_{d,1}(D)$ identifies with a $(d-1)$-dimensional daisy whose  cardinality is larger than $(\ell-p)\ell^{d-2}$, while $S_{d,0}(T^*)$ identifies with  a $(d-1)$-dimensional daisy of cardinality $(\ell-p)\ell^{d-2}$ (and we use the fact that daisies are ordered by cardinality). This allows to conclude, together with \eqref{qt} and \eqref{bDQ}, that
\[
b(M^*)=(\ell-p)\ell^{d-2}+b( T^*)+ b(D) > (\ell-p)\ell^{d-2}+b(\overline T)+ b(\overline Q)=b(\overline M),
\]
contradicting the fact that $\overline M$ is a $EIP^d$ minimizer and thus concluding the proof. 
\end{proof}

The next lemma provides the lower bound. 

\begin{lemma}\label{lemma:lb}
Let $d\in\{2,3,\ldots\}$. Let $\ell\in\mathbb N$.
The configuration 
\[
P_{\ell,d,p}:=\{1,\ldots, \ell-p\}\times\{1,\ldots,\ell\}^{d-1}
\]
is an $EIP^d$ minimizer for  any $p\in\mathbb N$ such that $p\le \lfloor h_{\ell,d}\rfloor$.
\end{lemma}
\begin{proof}
The statement holds if $d=2$. Indeed, the configuration $\{1,\ldots,\ell-p\}\times\{1,\ldots,\ell \}$ is an $EIP^2$ minimizer for any $p\in\{1,\ldots,\lfloor\sqrt{\ell}\rfloor \}$ as shown in \cite[Lemma 4.1]{MPSS}. We include a short alternative argument here: Wlog assume that $p \ge 2$ (and $\ell \ge 4$) since otherwise the claim follows from $P_{\ell,2,p}$ being a daisy. Then $D = D^{(2)}_{a,b,c}$ with $a = \ell - \lceil \frac{p}{2} \rceil$, $b = \ell - \lfloor \frac{p}{2} \rfloor -1$ and $c = \ell - \lceil \frac{p}{2} \rceil (\lfloor \frac{p}{2} \rfloor + 1)$ is a two-dimensional daisy (see \eqref{eq:two-d-daisy}) with $p \ge 2$ guaranteeing $c \le a-1$ and $c \ge \ell - ((\frac{p}{2})^2 + \frac{p}{2}+1) \ge \ell - \frac{p^2}{2} -1 \ge 1$ as $p \le \sqrt{\ell}$. The assertion then follows from $\#D = \ell^2 - \ell p = \# P_{\ell,2,p}$ and $\Theta_2 (D) = 4 \ell - 2p = \Theta_2(P_{\ell,2,p})$. 

Let $d\ge 3$. We  prove the statement by induction on the dimension: we assume that $P_{\ell,d-1,p}$ is an $EIP^{d-1}$ minimizer for any $p\le \lfloor  h_{\ell,d-1}  \rfloor$ and we aim at showing that $P_{\ell,d,p}$ is  an $EIP^d$ minimizer for any $p\le\lfloor h_{\ell,d}\rfloor$.
Thanks to Lemma \ref{easy2}, it is enough to show that
\[
M_{\ell,d,p}:=\{1,\ldots, \ell-p\}\times\{1,\ldots, \ell+p\}\times\{1,\ldots,\ell\}^{d-2} 
\]
is an $EIP^d$ minimizer for any $p\le\lfloor h_{\ell,d}\rfloor$. In order to check this, we rearrange $M_{\ell,d,p}$, without losing bonds, to
\[
\widetilde M_{\ell,d,p}
:=(\{1,\ldots,\ell\}\times\{1,\ldots,\ell-p\}\times\{1,\ldots,\ell\}^{d-2})\;\cup\;
(\{1,\ldots,\ell-p\}\times\{\ell-p+1,\ldots,\ell\}\times\{1,\ldots,\ell\}^{d-2} ).
\]
From the latter configuration, for any $i=1,\ldots, p$ and any $j\in 1,\ldots, p-1$ we fill the  $(d-2)$-dimensional section 
\begin{equation*}
U^{i,j}:=\{\ell-p+i\}\times\{\ell-p+j\}\times\{1,\ldots,\ell\}^{d-2}
\end{equation*}
by recursively rigidly moving the $(d-2)$-dimensional section 
$$\{\ell-p-k + 1\}\times\{\ell\}\times\{1,\ldots, \ell\}^{d-2},\qquad k=1,\ldots, p(p-1)$$ 
and filling the sets $U^{i,j}$ following the order $(i,j) \prec_R (i',j')\iff$ [($j<j'$) or ($j=j'$ and $i<i'$)], 
thus recursively emptying a $(d-1)$ dimensional face of $\widetilde M_{\ell,d,p}$, so that we get, 
\[
Q_{\ell,d,p}:=(\{1,\ldots,\ell\}\times\{1,\ldots,\ell-1\}\times\{1,\ldots,\ell\}^{d-2}) \;\cup\;(\{1,\ldots,\ell-p^2\}\times\{\ell\}\times\{1,\ldots,\ell\}^{d-2}).
\] 
(This is possible since $p^2 \le \lfloor h_{\ell,d}\rfloor^2 < \ell$ for $d \ge 3$.) We notice that $Q_{\ell,d,p}$ is a rearrangement of $\widetilde M_{\ell,d,p}$, with the same number of bonds.
By Definition \ref{rectdaisy}, 
$$ Q_{\ell,d,p}\setminus (\{1,\ldots,\ell-p^2\}\times\{\ell\}\times\{1,\ldots, \ell\}^{d-2})$$ 
is (up to a coordinate relabeling) a perfect daisy. Therefore, $Q_{\ell,d,p}$ is an $EIP^d$ minimizer as soon as 
$$
\{1,\ldots,\ell-p^2\}\times\{1,\ldots, \ell\}^{d-2}
$$
is an $EIP^{d-1}$ minimizer for then this set can be replaced by a $(d-1)$-dimensional daisy in $S_{2,\ell}(\mathbb Z^d)$ of cardinality $(\ell-p^2) \ell^{d-2}$ without decreasing the total number of bonds. The resulting configuration is (up to coordinate relabeling) a $d$-dimensional daisy and, thus, an $EIP^d$ minimizer. Assuming $p\le\lfloor h_{\ell,d}\rfloor$, by the elementary inequality $\lfloor x\rfloor^2\le\lfloor x^2\rfloor$ and by Definition \ref{parameter} we obtain
\[
p^2\le\lfloor h_{\ell,d}\rfloor^2\le \lfloor h_{\ell,d}^2\rfloor=\lfloor h_{\ell,{d-1}}\rfloor,
\]
which allows to conclude, by the induction assumption, 
that 
$$
\{1,\ldots,\ell-p^2\}\times\{1,\ldots, \ell\}^{d-2}
$$
is indeed an $EIP^{d-1}$ minimizer.
Therefore $Q_{\ell,d,p}$, $\widetilde M_{\ell,d,p}$ and $M_{\ell,d,p}$ are $EIP^d$ minimizers, as desired, for any $p\le\lfloor h_{\ell,d}\rfloor$.
\end{proof}

We shall later need the following converse statement.

\begin{lemma}\label{converse}
Let $d\in\{2,3,\ldots\}$. Let $\ell\in\mathbb N$ and $j \in \{0, \ldots, d-1\}$.
The configuration 
\[ 
P_{\ell,j,d,2p}:=\{1,\ldots, \ell-2p\}\times\{1,\ldots,\ell+1\}^{j}\times\{1,\ldots,\ell\}^{d-1-j} 
\]
is not an $EIP^d$ minimizer if $p\in\mathbb N$ is such that $2p\ge 4^{c_d}\,  h_{\ell,d}$, where $c_d:=1-2^{1-d}$. 
\end{lemma}

\begin{proof} 
Let $\tilde{\ell} = \ell+1$ if $j \ge 1$ and $\tilde{\ell} = \ell$ in case $j=0$. The result is true if $d=2$, as a consequence of \cite[Lemma 4.1]{MPSS}. It also directly follows by comparing with $D = \{1, \ldots, \ell - p \} \times \{1, \ldots, \tilde{\ell} -p - 1\}$ which for $p \ge \sqrt{\ell}$ satisfies $\#D = (\ell - p)(\tilde{\ell} - p - 1) \ge (\ell - 2p) \tilde{\ell} = \#P_{\ell,j,2,2p}$ while $\Theta_2(D) = 2 \ell + 2 \tilde{\ell} - 4p - 2 < 2 \ell - 2 \tilde{\ell} - 4p = \Theta_2(P_{\ell,j,2,2p})$. We prove the statement by induction. We consider the following two subsequent, edge-perimeter preserving rearrangements of $P_{\ell,j,d,2p}$:
\begin{align*}
P'
  &=\big( \{1,\ldots, \ell-2p\}\times\{1,\ldots, \tilde{\ell}-p\}
    \cup\{\ell-2p+1,\ldots, \ell-p\}\times\{1,\ldots,\ell-2p\} \big) \times H, \\ 
P''
  &=\Big(\big( \{1,\ldots, \ell-2p\}\times\{1,\ldots,\tilde{\ell}-p-1\} 
    \cup \{\ell-2p+1,\ldots, \ell-p\}\times\{1,\ldots,\ell-p-1\} \big) \times H \Big)\\ 
  &\qquad \cup\;\Big( \{1,\ldots, \ell-2p-p(p-1)\}\times \{\tilde{\ell}-p\} \times H\Big), 
\end{align*}
where we have set $H = \{1,\ldots,\ell+1\}^{j-1}\times\{1,\ldots,\ell\}^{d-1-j}$ if $j \ge 1$ and $H = \{1,\ldots,\ell\}^{d-2}$ if $j = 0$.

Here $P''$ is obtained from $P'$ by successively moving $d-2$ dimensional slices similarly as in the proof of Lemma \ref{lemma:lb}. We may assume without loss of generality that $\ell-2p-p(p-1) \ge 1$ for otherwise this process would terminate with an empty layer at the level $\mathbb Z \times \{\tilde{\ell} - p\} \times \mathbb Z^{d-2}$, i.e., at some point we are moving the $(d-2)$-dimensional section $\{1\}\times\{\tilde{\ell}-p\}\times H$, which would be the only remaining set of points with second component equal to $\tilde{\ell}-p$, to a position $\{\ell-2p+i\}\times\{\ell-2p+j\}\times H$ for some $i\in\{1,\ldots,p\}, j\in\{1,\ldots p-1\}$. This would strictly increase the number of bonds, which directly shows that $P'$ and thus $P_{\ell,j,d,2p}$ cannot be $EIP^{d-1}$ minimizers. 

In particular, by Corollary \ref{coro1} $P''$ (and thus $P_{\ell,j,d,2p}$) is not an $EIP^d$ minimizer if its face 
\[ \{1,\ldots, \ell-2p-p(p-1)\}\times H\] 
is not an 
$EIP^{d-1}$ minimizer. We make use of the induction assumption: the configuration $\{1,\ldots, \ell-2q\}\times H$ is not an $EIP^{d-1}$ minimizer if $2q\ge 4^{c_{d-1}}h_{\ell,d-1}$.
Therefore, the face $ \{1,\ldots, \ell-2p-p(p-1)\}\times H$ is not an $EIP^{d-1}$ minimizer (and thus $P_{\ell,j,d,2p}$ is not an $EIP^d$ minimizer), if 
\begin{equation}\label{kk}
  p(p+1) \ge 4^{c_{d-1}}h_{\ell,d-1}.
\end{equation}
The latter is  implied by $2p\ge 4^{c_d}h_{\ell,d}$: indeed, since $c_{d-1} + 1 = 2 c_d$ and $h_{\ell,d-1}=h_{\ell,d}^2$, we have
\[ (2p)^2/4 
   \ge (2^{2c_d} h_{\ell,d})^2/4 
   = 2^{2c_{d-1}+2} h_{\ell,d}^2/4 
   = 4^{c_{d-1}}h_{\ell,d-1}, \]
which readily implies \eqref{kk}.
 Therefore, if $2p\ge 4^{c_d}h_{\ell,d}$, we obtain that $P_{\ell,j,d,2p}$ is not an $EIP^d$ minimizer. 
\end{proof}

\section{Upper bound}\label{upper}

We introduce the notion of defects of a daisy, which is crucial for the  rearrangement procedures that will lead to the proof of the upper bound. In the following definition, we will consider a $d$-dimensional daisy $P=P^{(d)}\cup\ldots\cup P^{(1)}$. In order to define defects of lower-dimensional layers, given $m\in\{2,\ldots, d\}$, we recall that the set $P^{(m-1)}\cup\ldots\cup P^{(1)}$ is a copy of an $(m-1)$-dimensional daisy, through the identification provided by {\rm Proposition \ref{characterization}}.

\begin{definition}[Defects]\label{defect}
Let $P=P^{(d)}\cup P^{(d-1)}\cup\ldots\cup P^{(1)}$ be a $d$-dimensional daisy. 
\begin{itemize}
\item[i)]
 Let $R$ be a $d$-dimensional perfect daisy. We say that $P$ has a ($(d-1)$-dimensional) {\it defect} with respect to $R$ if 
a $(d-1)$-dimensional nonempty section $S=S_{s,j}(R)$ of $R$ (see Definition \ref{sect}) exists such that $\mathrm{dist}(S,P)=1$. In such case, the set $D:=\{y\in S:\mathrm{dist}(y, P)=1\}$ is the defect.
\item[ii)] 
Given $m\in\{2,\ldots, d\}$,
we say that the $(m-1)$-dimensional daisy $P^{(m-1)}\cup\ldots\cup P^{(1)}$ has an ($(m-2)$-dimensional) {\it defect} with respect to $P^{(m)}$ if it has a defect, according to  point i), with respect to the $(m-1)$-dimensional perfect daisy $$\qquad\quad Q^{(m-1)}:=\{1,\ldots, p_1^{(m)}\}\times\ldots\times\{1,\ldots, p_{z_{m}-1}^{(m)}\}\times \{1,\ldots, p_{z_{m}+1}^{(m)} \}\times\ldots\times\{1,\ldots, p_{m}^{(m)}\},$$ where $\{p_1^{(m)},\ldots, p_m^{(m)}\}$ are the coefficients of the perfect $m$-dimensional daisy $P^{(m)}$ and
 $z_m$ is the corresponding  value-change position, see Definition \ref{rectdaisy}. 
\end{itemize}
\end{definition}

\begin{remark}\label{defectremark}\rm We note that a $d$-dimensional daisy $P=P^{(d)}\cup\ldots\cup P^{(1)}$ has a defect with respect to the perfect  $d$-dimensional daisy $R$ if and only if $R\supsetneqq Q$, where $Q$ is the smallest perfect $d$-dimensional daisy such that $P\subseteq Q$. 
In particular, $Q$ also has  a defect with respect to  $R$. Moreover, the definition of daisy implies that if  $P^{(1)}\neq\emptyset$, then $P^{(1)}$ has necessarily a defect with respect to $P^{(2)}$ (we stress that by a ($0$-dimensional) defect for $P^{(1)}$ wrt $P^{(2)}$ we just mean a point). More generally, if $P^{(m-1)}\neq\emptyset$ and $P^{(m-2)}=\emptyset$, then $P^{(m-1)}$ has a defect wrt $P^{(m)}$. In particular, if $P=P^{(d)}\cup\ldots\cup P^{(1)}$ is a $d$-dimensional daisy and it is not perfect, then there exists $m\in \{2,\ldots, d\}$ such that $P^{(m-1)}\cup\ldots\cup P^{(1)}$ is not empty and has a defect with respect to $P^{(m)}$.
\end{remark}

Following Definition \ref{defect}, the first properties of defects are contained in the following

\begin{proposition}\label{defectpro}
If $P^{(m-1)}\cup\ldots\cup P^{(1)}$  has a defect with respect to $P^{(m)}$, then the defect contains a set $F$ which is a copy of the smallest $(m-2)$-dimensional face of $P^{(m-1)}$, and any point of $F$ has distance $1$ from $P^{(m-1)}$.
\end{proposition}

\begin{proof}
By assumption, $P^{(m-1)}\cup\ldots\cup P^{(1)}$ has a defect wrt the perfect $(m-1)$-dimensional daisy
$Q:=\{1,\ldots, p_1^{(m)}\}\times\ldots\times\{1,\ldots, p_{z_{m}-1}^{(m)}\}\times \{1,\ldots, p_{z_{m}+1}^{(m)} \}\times\ldots\times\{1,\ldots, p_{m}^{(m)}\}$. By Remark \ref{defectremark}, also the smallest perfect $(m-1)$-dimensional daisy $\hat Q$ containing $P^{(m-1)}\cup\ldots\cup P^{(1)}$ is strictly contained in $Q$ and has a defect wrt to $Q$. If $\tilde Q$ is the perfect $(m-1)$-dimensional daisy that follows $\hat Q$ in the order $\prec$, then the set $\tilde Q\setminus \hat Q$ is contained in (a section of) $Q$. Moreover, $\tilde Q\setminus \hat Q$ contains a set $F$ with the desired properties. 
More explicitly we define $F$ as follows.  Suppose $P^{(m-1)}$ is the perfect daisy $\{1,\ldots, t+1\}^j\times\{1,\ldots, t\}^{m-1-j}$ for suitable $t\in\mathbb N$ and $j\in\{0,\ldots, m-2\}$. If $P^{(m-2)} = \emptyset$, then $\hat{Q} = P^{(m-1)}$, $\tilde{Q} = \{1,\ldots, t+1\}^{j+1}\times\{1,\ldots, t\}^{m-2-j}$, and we set 
\[ F := 
   \begin{cases}
      \{1,\ldots, t+1\}^{j-1}\times\{1,\ldots, t\}\times\{t+1\}\times\{1,\ldots, t\}^{m-2-j} & \mbox{if } j \ge 1, \\ 
      \{t+1\}\times\{1,\ldots, t\}^{m-2} & \mbox{if } j =0. 
   \end{cases} \] 
In case $P^{(m-2)} \neq \emptyset$ (in particular $m \ge 3$), and so $\hat{Q} = \{1,\ldots, t+1\}^{j+1}\times\{1,\ldots, t\}^{m-2-j}$ and 
\[ \tilde{Q} = 
   \begin{cases}
       \{1,\ldots, t+1\}^{j+2}\times\{1,\ldots, t\}^{m-3-j} & \mbox{if } j \le m-3, \\
       \{1,\ldots, t+2\}\times\{1,\ldots, t+1\}^{m-2} & \mbox{if } j = m-2, 
   \end{cases} \] 
we set 
\[ F :=
   \begin{cases}
       \{1,\ldots, t\}\times \{t+1\}\times\{1,\ldots, t\}^{m-3} & \mbox{if } j = 0, \\  
       \{1,\ldots, t+1\}^{j-1}\times\{1,\ldots, t\}^2\times\{t+1\}\times\{1,\ldots,t\}^{m-3-j} & \mbox{if } 1 \le j \le m-3, \\ 
       \{t+2\}\times \{1,\ldots,t+1\}^{m-3}\times\{1,\ldots, t\} & \mbox{if } j = m-2. 
   \end{cases} \] 
We see that $F$ is a copy of $\{1,\ldots,t+1\}^{j-1}\times \{1,\ldots, t\}^{m-1-j}$ which is a smallest $(m-2)$-dimensional face of $P^{(m-1)}$ and that any point of $F$ has distance $1$ from $P^{(m-1)}$.
\end{proof}

A stronger statement holds: 

\begin{proposition}\label{defectpro2}
If $P^{(m-1)}\cup\ldots\cup P^{(1)}$  has a defect with respect to $P^{(m)}$ (or in general with respect to a perfect $(m-1)$-dimensional daisy), then the defect contains  a copy of $P^{(m-2)}\cup\ldots\cup P^{(1)}$.
\end{proposition}

\begin{proof}
By its definition, a defect is contained in an $(m-2)$-dimensional hyperplane that has distance $1$ from one of the lateral faces $L$ of  $P^{(m-1)}\cup\ldots\cup P^{(1)}$ (cf.\ Definition \ref{faces}) and it is made by all the points in such hyperplane whose distance from $L$ is $1$. Since $L$ identifies with an $(m-2)$-dimensional daisy by Proposition \ref{daisysection}, and since daisies are ordered by cardinality (see Theorem \ref{unique}), it is enough to show that $\#L\ge \#(P^{(m-2)}\cup\ldots\cup P^{(1)})$. 

Through the rest of the proof we make use of the notation
\[
P_1:= P^{(m-1)}\cup\ldots\cup P^{(1)},\qquad 
P_2:= P_1 \setminus P^{(m-1)}, 
\]
so that $P_2$ identifies with the $(m-2)$-dimensional daisy $P^{(m-2)}\cup\ldots\cup P^{(1)}$. Let $(p_1^{(m-1)},\ldots, p_{m-1}^{(m-1)})$ be the coefficients of the perfect $(m-1)$-dimensional daisy $P^{(m-1)}$ and let $z_{m-1}$ be the corresponding value-change position. By the definition of a daisy, we have
$P_2\subsetneqq Z$, where 
$$ Z:= 
   \{1,\ldots,p_1^{(m-1)}\}\times\ldots\times\{1,\ldots, p_{z_{m-1}-1}^{(m-1)}\}\times\{p_{z_{m-1}}^{(m-1)}+1\}\times\{1,\ldots, p_{z_{m-1}+1}^{(m-1)}\}\times\ldots\times\{p_{m-1}^{(m-1)}\},$$
and $P_2$ coincides with the  ($(m-2)$-dimensional) lateral face of $P_1$ that is made by all those points $z$ of $P_1$ whose $(z_{m-1})$-th component is $p_{z_{m-1}}^{(m-1)}+1$.
 If $L=P_2$ we are done, therefore from now we assume $L\neq P_2$. 
We notice that being $L$ another lateral face of $P_1$, we have
\[
L\setminus P_2=\{1,\ldots, p_1^{(m-1)}\}\times\ldots\times\{p_{j}^{(m-1)}\}\times\ldots\times\{1,\ldots, p_{m-1}^{(m-1)}\}
\]
for some $j\in\{1,\ldots, m-1\}\setminus\{z_{m-1}\}$, hence
\begin{equation}\label{eq1}
\#(L\setminus P_2)=\prod_{i\in\{1,\ldots, m-1\}\setminus\{j\}} p_i^{(m-1)}.
\end{equation}
Let  $W:=\{y=(y_1,\ldots, y_{m-1})\in \mathbb N^{m-1} : y_{z_{m-1}}= p_{z_{m-1}}^{(m-1)}+1,\ y_j=p_j^{(m-1)}\}$. Since $L$ coincides with the set of all the points of $P_1$ whose $j$-th coordinate is $p_{j}^{(m-1)}$ and since $P_2\subsetneqq Z$, we have
\begin{equation}\label{eq2}
P_2\setminus L= P_2\setminus W\subseteq Z\setminus W.
\end{equation}
But we notice that
\begin{equation}\label{eq3}
\#(Z\setminus W)=(p_j^{(m-1)}-1)\prod_{i\in\{1,\ldots, m-1\}\setminus\{z_{m-1},j\}} p_i^{(m-1)}.
\end{equation}
Thanks to \eqref{eq1}, \eqref{eq2} and \eqref{eq3}, we obtain
\[\begin{aligned}
\#P_2-\# L&=\#(P_2\setminus L)-\#(L\setminus P_2)\le \#(Z\setminus W)-\#(L\setminus P_2)\\
&= (p_j^{(m-1)}-1)\prod_{i\in\{1,\ldots, m-1\}\setminus\{z_{m-1},j\}} p_i^{(m-1)}\;\;-\prod_{i\in\{1,\ldots, m-1\}\setminus\{j\}} p_i^{(m-1)}\\&=(p_j^{(m-1)}-1-p_{z_{m-1}}^{(m-1)})\prod_{i\in\{1,\ldots, m-1\}\setminus\{z_{m-1},j\}}p_i^{(m-1)}\le 0,
\end{aligned}
\]
where the last inequality is due to the fact that $p_j^{(m-1)}-p_{z_{m-1}}^{(m-1)}\in\{-1,0,1\}$, by the definition of a daisy.
\end{proof}

\begin{definition}[Defect filling] \label{filling} 
Let $P=P^{(d)}\cup P^{(d-1)}\cup\ldots\cup P^{(1)}$ be a $d$-dimensional daisy. Let $m\in\{2,\ldots, d\}$.
Suppose that $D$ is a defect of the $(m-1)$-dimensional daisy $P^{(m-1)}\cup\ldots\cup P^{(1)}$ wrt $P^{(m)}$ (resp. wrt  a perfect $(m-1)$-dimensional daisy) according to point ii) of Definition \ref{defect} (resp. according to point i) of Definition \ref{defect}).
The {\it defect is filled} if a new configuration $P'_{m-1}$ is obtained from $P^{(m-1)}\cup\ldots\cup P^{(1)}$ by adding a nonempty subset $D'$ of $D$. The construction of $P'_{m-1}$ from $P^{(m-1)}\cup\ldots\cup P^{(1)}$ is therefore called a {\it defect filling}. Notice that each point of $D'$ has one and only one bond with $P^{(m-1)}\cup\ldots\cup P^{(1)}$. 
\end{definition}

\begin{definition}[Minimal rectangle]\label{mr}
Let $d\in\mathbb N$. Let $C\subset \mathbb N^d$ be a finite set. We define the {\it minimal rectangle} of $C$ as the smallest subset $R(C)$ of $\mathbb N^d$ such that $C \subseteq R(C)$ and such that $R(C)= x_0 + \{1,\ldots, a_1\}\times\ldots\times \{1,\ldots, a_d\}$ for some $x_0 \in \mathbb Z^d$ and $a_1,\ldots, a_d\in\mathbb N$.
\end{definition}

We are ready for the proof of the key statement.

\begin{lemma}\label{key}
Let $d\in\{2,3,\ldots\}$.
Let $C$ be an $EIP^d$ minimizer with minimal rectangle $R(C)$ according to {\rm Definition \ref{mr}}, and assume wlog that $x_0 = 0$ and $a_d\ge a_j$ for any $j=1,\ldots, d$.
Then there exists another $EIP^d$ minimizer $\bar C$ such that $\#C=\#\bar C$ and 
\begin{equation}\label{quasi}
\bar C =\{1,\ldots,\ell_1\}\times\ldots\times\{1,\ldots,\ell_{d-1}\}\times\{1,\ldots,a_d-1\}\cup F_1\cup F_2, 
\end{equation}
where $(\ell_1, \ldots, \ell_{d-1})$ is a $DO1$-tuple, $F_1$ is (a translate of) a $(d-1)$-dimensional daisy that is contained in the hyperplane $\{x \cdot \mathbf e_d=a_d\}$ and $F_2$ is a configuration contained in the hyperplane $\{ x\cdot \mathbf e_j=\ell_{j}+1\}$ for some $j\in\{1,\ldots, d-1\}$. 
\end{lemma}
\begin{proof}
Since $C$ is an $EIP^d$ minimizer, we may assume that it contains a point of the form $(i_1,\ldots, i_d)$ for any $i_d=1,\ldots, a_d$. Let $C'$ be the decreasing rearrangement of $C$ in the direction $\mathbf e_d$, see Proposition \ref{rear}. In particular, for any $j=1,\ldots, a_d$, we denote by $P_j$ the section $S_{d,j}(C')$ of $C'$ (see Definition \ref{sect}) and we say that  $P_j$ is the $j$-level of $C'$. We notice that the $j$-level $P_j$ identifies  with a $(d-1)$-dimensional daisy for any $j=1,\ldots,a_d$ and we have $P_j\subseteq P_{j-1}$ for any $j\in\{2,\ldots, a_d\}$, as a byproduct of the rearrangement definition.

We assume that $P_1$ is a $(d-1)$-dimensional perfect daisy (we shall get rid of this assumption at the end of the proof). We will show that, whenever the inclusion $P_j\subset P_1$ is strict (for some $j=2,\ldots, a_d-1$), then it is possible to move a point from the $a_d$-level  to the $j$-th level, obtaining another $EIP^d$ minimizer.
Therefore, the major issue is to show that this is possible without losing bonds. 

Suppose  that $j$ is the minimal natural number such that the inclusion $P_j\subset P_1$ is strict (in particular, $P_{j-1}=P_1$).
We denote by $Q$ the $(d-1)$-dimensional daisy at the $j$-level and by $\hat Q$ the $(d-1)$-dimensional daisy at the the top level $a_d$.
We introduce the usual daisy notation 
\[
P_j=Q=Q^{(d-1)}\cup Q^{(d-2)}\cup\ldots \cup Q^{(1)},\qquad 
P_{a_d}=\hat Q=\hat Q^{(d-1)}\cup\hat Q^{(d-2)}\cup\ldots\cup\hat Q^{(1)}.
\]
We also denote by $p_i^{(k)}$ and $\hat p_i^{(k)}$ the coefficients of such daisies from Proposition \ref{characterization}.
Recalling that daisies are identified by their cardinality (see Theorem \ref{unique}), we have $\hat Q\subseteq Q$, 
and then we split the proof in the following two possible cases:
\medskip 

\noindent \underline{Case 1: There exists $m\in\{1,\ldots, d-1\}$ such that for some $i\in\{1,\ldots, m\}$ there holds $p_i^{(m)}<\hat p_i^{(m)}$}. 

In this case, let $\bar m$ be the maximal of such $m$'s, so
 that
 \begin{equation}
\label{null} 
p_i^{(\bar m)}<\hat p_i^{(\bar m)}\;\mbox{ for some $i\in\{1,\ldots,\bar m\}$ }
\end{equation}
and
\begin{equation}\label{eins}
p_i^{(\bar m+1)}\ge\hat p_i^{(\bar m+1)}\;\mbox{ for all $i\in\{1,\ldots, \bar m+1\}$}.
\end{equation}
Note that $\bar{m} \leq d-2$. By the monotonicity of the sequences $\{1,\ldots, \bar{m}\}\ni i\mapsto p_i^{(\bar{m})}$ and $\{1,\ldots, \bar{m}\}\ni i\mapsto \hat p_i^{(\bar{m})}$ and the fact that their oscillation is at most $1$ (see Definition \ref{rectdaisy} and Proposition \ref{characterization}), we get
\begin{equation}\label{zwei}
p_i^{(\bar{m})}\le \hat p_i^{(\bar{m})} \;\mbox{ for all $i\in\{1,\ldots, \bar{m}\}$}.
\end{equation}
We consider the following two sets, obtained from $P_j$ and $P_{a_d}$ by exchanging the layers from $\bar m$ to $1$:
\[
\widetilde{P_j}=Q^{(d-1)}\cup\ldots\cup Q^{(\bar m+1)}\cup \hat Q^{(\bar m)}\cup\ldots\cup\hat Q^{(1)},
\]
\[
\widetilde{P_{a_d}}= \hat Q^{(d-1)}\cup\ldots\cup \hat Q^{(\bar m+1)}\cup  Q^{(\bar m)}\cup\ldots\cup Q^{(1)}.
\]
We claim that in view of Proposition \ref{characterization} these two new configurations are both daisies. Indeed, the claim is obvious for $\widetilde{P_{a_d}}$, since \eqref{zwei} implies $Q^{(\bar m)}\subset \hat Q^{(\bar m)}$ (and by \eqref{null} the inclusion is strict).
On the other hand in order to see that $\widetilde{P_j}$ is a daisy, we need to check that the sequence $i\mapsto p_i^{(\bar m+1)}$ is larger than the sequence $i\mapsto \hat p_i^{(\bar m)}$ in the sense of Definition \ref{larger}. But this is a direct consequence of \eqref{eins} and $(\hat p_1^{(\bar m+1)}, \ldots, \hat p_{\bar m+1}^{(\bar m+1)}) \sqsupset (\hat p_1^{(\bar m)}, \ldots, \hat p_{\bar m}^{(\bar m)})$. Therefore, $\widetilde{P_j}$ and  $\widetilde{P_{a_d}}$ satisfy all the assumptions in Definition \ref{dai} and the claim follows. We now consider the new configuration that arises from $C'$ by substituting $P_{a_d}$ with $\widetilde{P_{a_d}}$ and $P_j$ with $\widetilde{P_j}$. It has the same cardinality as $C'$ but a smaller upper face since \eqref{null} and \eqref{zwei} imply $\#\widetilde{P_j}>\#P_j$. In fact, it is also an $EIP^{d}$ minimizer, as desired, because the total number of bonds does not change: For the bonds perpendicular to $\mathbf e_d$ we have 
\begin{align*} 
  &b(\widetilde{P_j}) + b(\widetilde{P_{a_d}}) \\ 
  &\quad = b(Q^{(d-1)}\cup\ldots\cup Q^{(\bar m+1)}) + b(\hat Q^{(\bar m)}\cup\ldots\cup\hat Q^{(1)}) + (d-1-\bar m) \# \hat Q^{(\bar m)}\cup\ldots\cup\hat Q^{(1)} \\
	&\quad \quad + b(\hat Q^{(d-1)}\cup\ldots\cup \hat Q^{(\bar m+1)}) + b(Q^{(\bar m)}\cup\ldots\cup Q^{(1)}) + (d-1-\bar m) \# Q^{(\bar m)}\cup\ldots\cup Q^{(1)} \\ 
	&\quad =b(P_j) + b(P_{a_d}). 
\end{align*}
Also the number of bonds in the $\mathbf e_d$ direction is conserved as lost bonds between the $a_d$ and $a_{d}-1$ layer are restored as new bonds between the $j$-th layer and the perfect daisy $P_{j-1}$. 
\medskip 

\noindent \underline{Case 2: For all  $m\in\{1,\ldots, d-1\}$, the inequality $p_i^{(m)}\ge \hat p_i^{(m)}$ holds for all $i\in\{1,\ldots, m\}$}.

This means that for any $m\in\{1,\ldots, d-1\}$,  $\hat Q^{(m)}$ is a subset (possibly not strict) of $Q^{(m)}$. In order to show that it is possible to move  points from the $a_d$-level  to the $j$-th level
we provide an iteration algorithm.

Before introducing the full algorithm, let us start by discussing the basic instance. If $Q$ has a defect with respect to the perfect $(d-1)$-dimensional daisy $P_{j-1}$, and if $\hat Q$ is not perfect, i.e.\ if $\hat Q^{(d-2)}\neq\emptyset$, we remove $\hat Q^{(d-2)}\cup\ldots\cup \hat Q^{(1)}$ from the top layer and use it to fill the defect (see Definition \ref{filling}). Indeed, $\hat Q^{(m)} \subseteq Q^{(m)}$ for all $m$ implies that $\hat Q^{(d-2)}\cup\ldots\cup \hat Q^{(1)}$ is, after a rigid motion, a subset of the defect thanks to Proposition \ref{defectpro2}. Thereby the total number of bonds is unchanged, as all the bonds of $\hat Q^{(d-2)}\cup\ldots\cup \hat Q^{(1)}$  with $\hat Q^{(d-1)}$ (whose number is $n:=\#(\hat Q^{(d-2)}\cup\ldots\cup \hat Q^{(1)}) $)
 are restored as bonds with $Q$. Also, the $n$ bonds of $\hat Q^{(d-2)}\cup\ldots\cup \hat Q^{(1)}$  with $P_{a_d-1}$
 are all replaced with bonds connecting to the larger daisy $P_{j-1}$.

Let us now introduce the algorithm. Starting from $k=d-1$ and decreasing $k\ge2$, we perform the following iteration procedure: 
\begin{center}{\it
if $ Q^{(k)}\cup\ldots\cup Q^{(1)}$ does not have a defect with respect to $Q^{(k+1)}$ and $\hat Q^{(k-1)}\neq \emptyset $, \\ proceed to check $Q^{(k-1)}\cup\ldots\cup Q^{(1)}$ and $\hat Q^{(k-2)}$.}
\end{center}
Here $Q^{(d)}$, which occurs if $k = d-1$, is understood as $P_{j-1}$. We have three possible situations:
\begin{itemize}
\item[A)] The procedure does not stop and reaches $k=2$, with no defects in $Q^{(2)}\cup Q^{(1)}$ (wrt $Q^{(3)}$) and $\hat Q^{(1)}\neq\emptyset$.
In such case $Q^{(1)}$ is nonempty and has a ($0$-dimensional) defect, see Remark \ref{defectremark}. Therefore we take a corner point from $\hat Q$ which has $d$ bonds to other points, to fill this defect without reducing the total number of bonds.

\item[B)] The procedure stops at some $k\ge 2$ with a defect in $ Q^{(k)}\cup\ldots\cup Q^{(1)}$ (wrt $Q^{(k+1)}$) and nonempty $\hat Q^{(k-1)}$.
As $ \hat Q^{(k-1)}\cup\ldots\cup \hat Q^{(1)}$ is nonempty and $p_i^{(m)} \ge \hat p_i^{(m)}$ for all $i \in \{1, \ldots, m\}$, $m \in \{1, \ldots, k-1\}$, we can proceed as above to fill a defect of $Q^{(k)}$ with a copy of $\hat Q^{(k-1)}\cup\ldots\cup \hat Q^{(1)}$. Here, removing such a portion from the top layer destroys $\bar n(d-k+1)$ bonds, where $\bar n = \# \hat Q^{(k-1)}\cup\ldots\cup \hat Q^{(1)}$, while filling the defect restores the same number of bonds. 

\item[C)] The procedure stops at some $k\ge 2$ with $\hat Q^{(k-1)}=\emptyset$. 

We will define $\hat S_{k-1}$ as one of the smallest $(k-1)$-dimensional faces of $\hat Q^{(k)}$. $\hat S_{k-1}$ identifies with  a $(k-1)$-dimensional daisy thanks to Proposition \ref{daisysection}, and in fact with a perfect daisy since $\hat Q^{(k)}$ is a perfect daisy. More precisely and more generally, for the perfect $k$-dimensional daisy $\hat Q^{(k)}$ and  for $j\in\{0,\ldots, k\}$ we define $\hat S_{k-j}$ as the set that is obtained by taking all the points $z=(z_1,\ldots, z_k)\in\hat Q^{(k)}$ and by freezing the first $j$ entries of $z$ to their maximal value. Then $\hat S_{k-j}$ is a perfect $(k-j)$-dimensional daisy and a copy of the smallest $(k-j)$ dimensional face of $\hat Q^{(k)}$ (in particular, $\hat S_k=\hat Q^{(k)}$ and $\hat S_0$ is a single corner point of $\hat Q^{(k)}$). We stress that each point of $\hat S_{k-1}$ has one bond with a point of $\hat Q^{(k)}\setminus \hat S_{k-1}$, unless $\hat Q^{(k)}$ is made of a single point (which is the only situation yielding $\hat S_{k-1}=\hat Q^{(k)}$). 
\begin{itemize}
\item[C1)]
If there are defects in $Q^{(k)}\cup\ldots\cup Q^{(1)}$, by Proposition \ref{defectpro} the defect contains a copy of  $S_{k-1}$, the smallest $(k-1)$-dimensional face of $Q^{(k)}$. But $\hat Q^{(k)}\subseteq Q^{(k)}$ implies $\hat S_{k-1}\subseteq S_{k-1}$. Therefore we can move $\hat S_{k-1}$ to fill the defect, as soon as $\hat Q^{(k)}$ is not made by a single point, since each of the bonds of $\hat S_{k-1}$ with $\hat Q^{(k)}\setminus \hat S_{k-1}$ is restored as a bond with $Q^{(k)}\cup\ldots\cup Q^{(1)}$ through this defect filling (Definition \ref{filling}). Also the lost bonds with $\hat Q^{(d-1)}\cup\ldots\cup \hat Q^{(k+1)}$ and with the $a_d-1$ layer are restored. Now note that $\#\hat S_{k-1}=\#\hat Q^{(k)}=1$ is not possible, since otherwise the defect filling would increase the number of bonds and contradict the minimality of $C'$. In particular, this defect filling does not exhaust $\hat Q^{(k)}$.

\item[C2)]
Assume now there are no defects in $Q^{(k)}\cup\ldots\cup Q^{(1)} = Q^{(k)}\cup\ldots\cup Q^{(h)}$, where $h\in \{1,\ldots k-1\}$ is such that $Q^{(h)} \neq \emptyset$ and  and $Q^{(h-1)} = \emptyset$. 

Suppose first that $\hat S_h \subseteq Q^{(h)}$. Since $Q^{(h)}$ has a defect due to Remark \ref{defectremark}, by Proposition \ref{defectpro} this defect contains a copy of the smallest $(h-1)$-dimensional face of $Q^{(h)}$. Since $\hat S_h \subseteq Q^{(h)}$, then the defect also contains a copy of $\hat S_{h-1}$. We can thus remove $\hat S_{h-1}$ from the top level and use it to fill the defect. Similarly as above,  each point of $\hat S_{h-1}$ has one bond with $\hat S_h\setminus \hat S_{h-1}$ unless the latter is empty, and these bonds are restored as bonds with $Q^{(h)}$ through the defect filling. Again, $\hat S_{h}\setminus\hat S_{h-1}=\emptyset$ is not possible (because the defect filling would create new bonds, contradicting the minimality of $C'$), so that $\hat Q^{(k)}$ is not exhausted.

Now suppose that, on the contrary, $\hat S_h \supsetneqq Q^{(h)}$. Since $Q^{(k)} \supseteq \hat Q^{(k)} = \hat S_k$, there is an index $i \in \{ h, \ldots, k-1 \}$ such that $Q^{(k)} \supseteq \hat S_k, \ldots, Q^{(i+1)} \supseteq \hat S_{i+1}$ but $Q^{(i)} \subsetneqq \hat S_i$. (Recall that daisies are totally ordered by inclusion) As $\hat S_i$ is a perfect daisy, we also have $Q^{(i)} \cup \ldots \cup Q^{(h)} \subsetneqq \hat S_i$. Since $Q^{(i+1)} \supseteq \hat S_{i+1}$ it is then possible to exchange the two sets $Q^{(i)} \cup \ldots \cup Q^{(h)}$ and $\hat S_i$ without changing the total number of bonds: indeed, we remove these two sets from their position by rigidly moving $\hat S_i$ into the $i$-dimensional affine hyperplane that was occupied by $Q^{(i)} \cup \ldots \cup Q^{(h)}$ in such a way that all the bonds that have deleted while detaching $\hat S_i$ are restored as bonds with $P_{j-1}$ and with $Q^{(d-1)} \cup \ldots \cup Q^{(i+1)}$, and similarly by moving $Q^{(i)} \cup \ldots \cup Q^{(h)}$ rigidly to a subset originally occupied by $\hat S_i$, restoring all bonds that have been deleted while detaching $Q^{(i)} \cup \ldots \cup Q^{(h)}$. 
\end{itemize}
\end{itemize}

We have shown that it is always possible to take points from the $a_d$-level to the $j$-level. Both in Case 1 and Case 2 above, the $a_d$-level is not  exhausted by this procedure. Indeed, in Case 1 we see that $\hat Q^{(d-1)}$ is left at the top level. Moreover, we have seen through the different instances of the algorithm in Case 2 that the top level is not exhausted. Therefore, we can repeat the procedure, and with a finite number of steps we reach a configuration of the form
\[
\{1,\ldots,\ell_1\}\times\ldots\times\{1,\ldots,\ell_{d-1}\}\times\{1,\ldots,a_d-1\}\cup F_1,
\]  
as desired. 

Let us conclude by generalizing the argument in case $P_1$ is not a perfect daisy. As $P_1=P_1^{(d-1)}\cup\ldots \cup P_1^{(1)}$, let us consider the set of points $H$ in $C'$ whose projection on $\{x\cdot\mathbf e_d=1\}$ belongs to $P_1 \setminus P_1^{(d-1)}$. Let $k\in\{1,\ldots, a_d\}$ denote the top level where points of $H$ are found. If $k\le a_d-2$,  then $P_{k}^{(d-1)}=P_1^{(d-1)}$ and $P_{k+1}\subseteq P_{k}^{(d-1)}$. Therefore we can proceed as before with $P_{k}$ in place of $P_1$. We obtain a configuration of the form \eqref{quasi} with $F_2=H$. If $k\in\{a_d-1,a_d\}$, then $C'$ is already of the form \eqref{quasi}, with the  $\ell_i$'s being the coefficients of the perfect daisy $P_1^{(d-1)}$.
\end{proof}

\begin{corollary}\label{cor:ad-est}
Let $C\subset \mathbb Z^d$ be an $EIP^d$ minimizer. Let $R(C)$ be the minimal rectangle and assume $x_0 = 0$. Let $a_d$ be the maximal edge of $R(C)$. Let $\ell = \ell_1$ from \eqref{quasi}. Then 
\[ a_d-\ell\le4^{c_d}h_{\ell,d} + 6. \] 
\end{corollary}
\begin{proof}
From Lemma \ref{key} we obtain $\bar C$ as in \eqref{quasi} with top layer $F_1$ and (possibly a) lateral face $F_2$ contained in $\{x\cdot \mathbf e_j=\ell_{j}+1\}$ for some $j\in\{1,\ldots, d-1\}$. Wlog we assume that $p = \lfloor \frac{a_d - \ell}{2} \rfloor \ge 3$. 

Throughout the proof, we perform transformations that delete and restore only the bonds in the directions $\mathbf e_j$ and $\mathbf e_d$. We first obtain another $EIP^d$ minimizer by cutting the entire block of points at the levels from $a_d-p+1$ to $a_d$ and paste it after a rigid motion  to the lateral face of $\bar C$ that is contained in the hyperplane $\{\mathbf e_j\cdot x=1\}$. In particular, we perform this rigid motion by letting the moved points from $F_2$  find their new positions at the first level, i.e., on the hyperplane $\{\mathbf e_d\cdot x=1\}$ and the points from $F_1$ on the hyperplane $\{\mathbf e_j\cdot x=1-p\}$. More precisely, any $x \in \bar C$ with $x_d \in \{a_d-p+1, \ldots, a_d\}$ is mapped to 
\[ (x_1,\ldots, x_{j-1}, a_d-p+1-x_d, x_{j+1}, \ldots, x_{d-1}, \ell_j+2-x_j). \]
This is possible without reducing the number of bonds since $\ell_j + 1 \le \ell + 1 \le a_d - p$. 
In this way, the obtained configuration $C'$ contains the set 
\[ Y := \prod_{i=1}^{j-1} \{1, \ldots, \ell_i\} \times \{-p+2, \ldots, 0\} \times \prod_{i=j+1}^{d-1} \{1, \ldots, \ell_i\} \times \{\ell_j+1\} \] 
but not the points above $Y$ in the $\mathbf e_d$ direction. Moreover, the top level of $C'$ is the level $a_d - p$, and precisely it is the set $\big( \prod_{i=1}^{d-1} \{1,\ldots, \ell_i\}\times\{a_d - p\} \big) \cup F_2^{a_d-p}$, where $F_2^{a_d-p}:=S_{d,a_d-p}(F_2)$. 

Let $k=p^2-3p$ if $\ell_j = \ell_{d-1}$ and $k=p^2-3p+1$ if $\ell_j = \ell_{d-1}+1$. We move points from the level $a_d - p$ to obtain another $EIP^d$ minimizer, whose upper face is 
\[ U := \prod_{i=1}^{j-1} \{1, \ldots, \ell_i\} \times \{k+1, \ldots, \ell_j\} \times \prod_{i=j+1}^{d-1} \{1, \ldots, \ell_i\} \times \{a_d-p\}. \] 
This is done, similarly to the constructions of Section \ref{lower}, by moving  $(d-2)$-dimensional faces of the top level (one by one): we remove $\{ x \in C' : x_j = i, \, x_d = a_d-p \}$ for $i=1,\ldots, k$, and place such $(d-2)$-dimensional layers at the positions 
\[ \prod_{i=1}^{j-1} \{1, \ldots, \ell_i\} \times \{-j_1\} \times \prod_{i=j+1}^{d-1} \{1, \ldots, \ell_i\} \times \{j_2\}, \] 
where $j_1 \in \{0,1,\ldots, p-2\}$ and $j_2\in\{\ell_j+2,\ldots, a_d-p-1\}$, which are the $(p-1)(a_d-p-\ell_j-2) \ge(p-1)(p-2)\ge k+1$ free positions above $Y$. (This is done, say, following the right-to-left lexicographic order of those $(j_1,j_2)$.) In doing so we fill $k$ of such free positions, and if $F_2^{a_d-p}\neq\emptyset$, we finally move it to fill the $(k+1)$-st position.

Since the upper face $U$ is necessarily an $EIP^{d-1}$ minimizer by Corollary \ref{coro1}, from Lemma \ref{converse} we infer 
\[ p^2-3p 
   \le 4^{c_{d-1}}\,h_{\ell_{d-1},d-1} 
   \le 4^{c_{d-1}}\,h_{\ell,d-1}, \]
which implies, by using the relations $c_{d-1}+1=2c_d$ and $h_{\ell,d-1}=h^2_{\ell,d}$,
\[ 2p 
   \le 3+\sqrt{9+4^{c_{d-1}+1} h_{\ell,d-1}}\le 5+ 2^{c_{d-1}+1}\sqrt{h_{\ell,d-1}}=4^{c_d}h_{\ell,d}+5,
\]
where we have also used the elementary inequality $3+\sqrt{9+x}\le 5+\sqrt{x}$, which holds for $x\ge 2$  (noticing that $4^{c_{d-1}+1}h_{\ell, d-1}\ge 4$ as $d\ge 2$). The result is proven. 
\end{proof}

\begin{proof}[Proof of Theorem \ref{thm:main}]
\noindent (i) Let $C$ be an $EIP^d$ minimizer with $\#C = n$. Wlog suppose and $R(C) = \{1, \ldots, a_1\} \times \ldots \times \{1, \ldots, a_d\}$ and $a_1, \ldots, a_{d-1} \le a_d$. By Lemma \ref{key} and Corollary \ref{cor:ad-est}, with $\ell = \ell_1$ from \eqref{quasi} we have $n = \ell^{d-1} a_d + O(\ell^{d-2} a_d)$ and $a_d-\ell\le 4^{c_d} h_{\ell,d} + 6$. In particular, $n = \ell^d + O(h_{\ell,d} \ell^{d-1})$. We also observe that \eqref{quasi} gives $n \ge (\ell-1)^d$. 

Now suppose there is an $i$ with $a_i \le \ell - 2 d 4^{c_d} h_{\ell,d}$. Then 
\begin{align*}
  n 
   \le \# R(C) 
  &\le (\ell - 2 d 4^{c_d} h_{\ell,d}) (\ell + 4^{c_d} h_{\ell,d} + 6)^{d-1} \\ 
  &= \ell^d (1 - 2 d 4^{c_d} h_{\ell,d} \ell^{-1}) (1 + (4^{c_d} h_{\ell,d} + 6) \ell^{-1})^{d-1}. 
\end{align*}
Using that $h_{\ell,d} \ell^{-1} \to 0$ as $n \to \infty$ and $(1 + (4^{c_d} h_{\ell,d} + 6) \ell^{-1})^{d-1} = 1 + (d-1) (4^{c_d} h_{\ell,d} + 6) \ell^{-1} + O((h_{\ell,d} \ell^{-1})^2)$, we find that for $n$ sufficiently large,  
\[ \ell^d (1 - d \ell^{-1}) 
   \le \ell^d (1 - \ell^{-1})^d 
   \le n 
   \le \ell^d (1 - 2 d 4^{c_d} h_{\ell,d} \ell^{-1}) (1 + d 4^{c_d} h_{\ell,d} \ell^{-1}) 
\]
and so 
$$ 1 - d \ell^{-1} \le 1 - d 4^{c_d} h_{\ell,d} \ell^{-1}, $$ 
contradicting $h_{\ell,d} \to \infty$ as $n \to \infty$. This shows that in fact $a_i \ge \ell - 2 d 4^{c_d} h_{\ell,d}$ for all $i$ if $n$ is large enough. 

As a consequence we have 
$$ \# R(C) \triangle \{1, \ldots, \ell\}^d 
   \le (\ell + 4^{c_d} h_{\ell,d})^d - (\ell - 2 d 4^{c_d} h_{\ell,d})^d 
   = O(h_{\ell,d}\ell^{d-1}). $$ 
From $n = \ell^d + O(h_{\ell,d} \ell^{d-1}) = \ell^d(1 +  O(h_{\ell,d} \ell^{-1}))$ and thus $\lfloor n^{1/d} \rfloor = \ell (1 +  O(h_{\ell,d} \ell^{-1})) = \ell + O(h_{\ell,d})$ we also obtain 
$$ \# W_n \triangle  \{1, \ldots, \ell\}^d 
   = O(h_{\ell,d}\ell^{d-1}). $$
So by the triangle inequality we get 
$$ \# W_n \triangle C 
   = O(h_{\ell,d}\ell^{d-1}) $$
as claimed. 
\smallskip

\noindent (ii) This follows directly from Lemma \ref{lemma:lb}. 
\end{proof}

\subsection*{Acknowledgements} E.M.\ acknowledges support from the MIUR-PRIN  project  No 2017TEXA3H and from the INdAM-GNAMPA 2019 project {\it ``Trasporto ottimo per dinamiche con interazione''}. Both authors wish to thank Paolo Piovano and Ulisse Stefanelli for interesting discussions on the subject of the paper. 


\end{document}